%% file: paper.tex
\documentclass[natbib]{svjour3}

\let\makeheadbox\relax

\usepackage{hyperref}

\usepackage{paper}

\ExplSyntaxOn
\input{saved.tex}
\ExplSyntaxOff

\title{Fair Allocation of Conflicting Items}
\author{Halvard Hummel \and Magnus Lie Hetland}
\institute{
    H. Hummel
    \and
    M. L. Hetland
    \at
    Norwegian University of Science and Technology, Norway,\\
    \email{halvard.hummel@ntnu.no, mlh@ntnu.no}
}
\journalname{}
\date{January 31, 2022}

\begin{document}
\maketitle

\begin{abstract}
    We study fair allocation of indivisible items, where the items are furnished
    with a set of \emph{conflicts}, and agents are not permitted to receive
    conflicting items. This kind of constraint captures, for example,
    participating in events that overlap in time, or taking on roles in the
    presence of conflicting interests. We demonstrate, both theoretically and
    experimentally, that fairness characterizations such as EF1, MMS and MNW
    still are applicable and useful under item conflicts. Among other existence,
    non-existence and computability results, we show that a
    $1/\Delta$-approximate MMS allocation for $n$ agents may be found in
    polynomial time when $n>\Delta>2$, for any conflict graph with maximum
    degree $\Delta$, and that, if $n > \Delta$, a 1/3-approximate MMS allocation
    always exists.
\end{abstract}

\section{Introduction}

We are interested in the problem of allocating a set of indivisible items among
a set of agents with additive valuations, and beyond finding an efficient
solution, where the total utility is high, we want the allocation to be fair, in
some sense---a problem that has been studied extensively in the last couple of
decades~\citep[]{Brandt:2016}. More recently, variations of this problem have
appeared, where the bundles of items allocated must conform to some
\emph{constraints}, meaning that not all allocations are \emph{feasible}. For
example, if the items are structured as a matroid, one may require that the set
of all allocated items form a basis~\citep{Gourves:2013}, or that each bundle be
an independent set of the matroid~\citep{Biswas:2018}. One might also partition
the items into different categories, and limit each agent to a certain number
from each~\citep{Biswas:2018}. Or the items may be the nodes of a graph, where a
natural requirement is for each bundle to be connected~\citep{Bouveret:2017}.

In this paper, we look at yet another form of constraint, where items may be in
\emph{conflict} with each other, and an agent may receive at most one of any two
conflicting items. This is a situation that may occur in many realistic
allocation scenarios. For example, such conflicts arise naturally in
scheduling problems where the items represent participation in some
activities---such as conference sessions or panels at a convention---with
limited seating. The activities are associated with time intervals, and no agent
may participate in two of them simultaneously. Thus when participants register
for a prioritized subset, any fair allocation must take care to avoid scheduling
conflicts. In another scenario, the items may represent sought-after
administrative positions in an academic institution, where the conflicts are
conflicts of interest. Positions should be allocated fairly among qualified
applicants, without anyone appointed to two positions where one has the power to
approve proposals submitted by the other, for example.

\paragraph{Our contributions.} We study the problem of fairly allocating
conflicting items, and present several new results in this setting. We map out,
both in broad strokes and for certain special cases, to what extent
envy-freeness up to one item (EF1) is guaranteed to exist, as well as when it is
guaranteed by maximum Nash welfare (MNW), and when it may be achieved in
polynomial time (\cref{sec:envyfreeness}). We adapt the randomized graph
coloring procedure of \citet{Pemmaraju:2008}, which is expected to approximate
any individual agent's proportional share to within a factor of $1-1/e$, and
guarantees that for large instances of certain kinds, the deterministic
proportionality guarantee does not fall far below this expectation
(\cref{sec:prop}). We then present a series of results on maximin share (MMS)
allocation with item conflicts (\cref{sec:mms}), with the main results being
that, if the maximum degree $\Delta$ of the conflict graph is lower than the
number of agents,\footnote{This restriction ensures that any partial allocation
may be completed without reassigning any items.}
\begin{inl}
\item there exists an $\alpha$-approximate MMS allocation, with $\alpha > 1/3$
    (\cref{thr:mms-alpha-existence}); and
\item an $\alpha$-approximate MMS allocation may be found in polynomial time,
    with $\alpha > 1/\Delta$ when $\Delta>2$ (\cref{thr:mms-approx-polytime}).
\end{inl}
Finally, we examine the behavior of various fairness measures in
\emph{practice}, that is, empirically studying how random allocation, EF1, MMS,
proportionality and MNW are affected
by the introduction of item conflicts, on \num{\ntotal} randomly generated
graphs (\cref{sec:practice}), with the main conclusion being that fairness is
largely unharmed, with random allocation improving as an MMS approximation,
all instances exhibiting EF1 and MMS, and MNW producing EF1 in almost all cases,
with a tight approximation of MMS.

\paragraph{Prior work.} Fair allocation of conflicting items has been studied
by~\citet{Chiarelli:2020}, who considered partial egalitarian (maximin)
allocations, i.e., allocations that maximize the value the worst-off agent
receives, but where some items may remain unallocated. In this paper, we study
conventional (complete) allocations, and focus on other fairness criteria such
as envy-freeness, maximin share guarantees and Nash welfare.
The
relation between our scenario and that of \citeauthor{Chiarelli:2020} depends on
the relationship of the maximum degree of the conflict graph, $\Delta$, and the
number of agents, $n$.  For the maximin criterion used by
\citeauthor{Chiarelli:2020}, there is no difference between allowing and
disallowing partial allocations when $n>\Delta$, which is also a precondition
for several of our results. Conversely, the hardness results of
\citeauthor{Chiarelli:2020} rely on graphs where $\Delta>n$, which means they do
not contradict our approximability result for MMS
(\cref{thr:mms-approx-polytime}). Without the fairness aspect, allocation with
item conflicts reduces to the much-studied problem of \emph{graph
coloring}~\citep{Lewis:2016}, with each item becoming a node, and each color an
agent.  Some work has already been done on equity in graph
coloring~\citep{Lih:2013}, but links to the burgeoning field of fair item
allocation seem so far to be missing from the literature.

\section{Preliminaries}
\label{sec:prelim}

We study the problem of fairly allocating a set of items among a set of agents,
where certain pairs of items are not allowed together. We call this
problem \emph{fair allocation of conflicting items}. In the following, for $k
\in \mathbb{Z}^+$, $[k]$ denotes the set $\{1, 2, \dots, k\}$.

\begin{definition}
    An \emph{instance} of the \emph{fair allocation of conflicting items}
    problem, or more concisely, a \emph{problem instance}, is a quadruple $(N,
    M, V, G)$, where
    \begin{itemize}
        \item $N$ is a set of $n$ \emph{agents};
        \item $M$ is a set of $m$ \emph{items};
        \item $V$ is a family of $n$ \emph{valuation functions},
            $v_i : 2^{M} \rightarrow \mathbb{R}_{\ge 0}$; and
        \item $G = (M, E)$ is an undirected \emph{conflict graph}.
    \end{itemize}
    Unless otherwise stated, we assume $N=[n]$ and $M=[m]$.\footnote{An
    exception is reduced instances where agents or items have been removed (cf.\
    \cref{sec:mms}).}
\end{definition}

\noindent
We say that two neighboring items in the conflict graph are \emph{conflicting
items} or, equivalently, \emph{in conflict}. Although more general valuations
are possible, we shall assume that all valuation functions are \emph{additive}.
For simplicity, we let both $v_{ij}$ and $v_i(j)$ denote $v_i(\{j\})$.

For an instance of the fair allocation of conflicting items problem, an
\emph{allocation}, $A = \langle A_1, A_2, \dots, A_{|N|} \rangle$, is an
$|N|$-partition of the set of items, assigning set $A_i$ to agent $i$. A set of
items $A_i \in A$ is called a \emph{bundle}. An $|N|$-partition of a strict
subset of the items is called a \emph{partial allocation}. An allocation that is
not partial is \emph{complete}.

We wish to find feasible allocations that are as fair as
possible. An allocation is said to be \emph{feasible} if no bundle contains a
pair of conflicting items. Note that any feasible allocation forms an
$|N|$-coloring of the conflict graph and that each bundle in a feasible
allocation is an independent set in the conflict graph. What constitutes a
\emph{fair} allocation is less clear-cut, and many characterizations exist.
In this paper, we consider the four fairness criteria of \emph{maximum Nash
welfare}, \emph{envy-freeness up to one good}, \emph{proportionality} and
\emph{maximin share guarantee}.

The Nash social welfare function, or \emph{Nash welfare}, is similar to a plain
utilitarian welfare, except that individual utilities are \emph{multiplied} to
produce an aggregate. For allocation without conflicts, maximizing the Nash
welfare leads to a good tradeoff between fairness and efficiency and guarantees
fulfillment or close approximation of several other fairness
criteria~\cite{Caragiannis:2019}.

\begin{definition}
    For a problem instance $(N, M, V, G)$, the \emph{Nash welfare} (NW) of a
    feasible allocation $A$ is given by
    \[
        \NW(A) = \left(\prod_{i \in N} v_i(A_i)\right)^{1/n}\mkern-4mu\eqdot
    \]
    $A$ is said to be a \emph{maximum Nash welfare} (MNW) allocation if there
    is no feasible allocation with a higher Nash welfare.
\end{definition}

\noindent
\emph{Envy-freeness} is a very natural criterion for fair allocation, which
requires that no agent be envious of any other agent. Envy-freeness is often
unobtainable when considering indivisible items. This is easily seen when
allocating a single item to two agents. Instead of full envy-freeness, we
consider the relaxation to \emph{envy-freeness up to one good}, introduced by
\citet{Budish:2011}. This fairness criterion instead requires that for any pair
of agents $i$ and $i'$, the bundle $A_{i'}$ contains an item so that $i$ would
not envy $i'$ if the item were removed from $A_{i'}$. More formally:

\begin{definition}
    For a problem instance $(N, M, V, G)$, a feasible allocation $A$ is said
    to be \emph{envy-free up to one good} (EF1) if for all $i, i' \in N$, where
    $A_{i'} \neq \emptyset$,

    \[
        v_i(A_i) \ge v_i(A_{i'}) - \max_{j \in A_{i'}}v_{ij}\eqdot
    \]
\end{definition}

\noindent
A different relaxation of envy-freeness is \emph{proportionality}, where agents
should receive at least their subjective fair of the total value available.
More formally:

\begin{definition}
    For a problem instance $(N, M, V, G)$, a feasible allocation $A$ is called
    \emph{proportional} if each agent $i$ assigns a value of at least
    $v_i(M)/|N|$ to its bundle.
\end{definition}

\noindent
Maximin share guarantee is another fairness criterion introduced by
\citet{Budish:2011}. Here, we want to guarantee each agent a bundle valued at
no less than what the agent would have been guaranteed if it were to
create a feasible allocation, but had to choose its own bundle last.

\begin{definition}
    For a problem instance $I = (N, M, V, G)$, an agent $i$'s \emph{maximin
    share} (MMS) is given by
    \[
        \mu_i^I = \max_{A \in \mathcal{F}} \min_{A_j \in A} v_i(A_j)\eqcomma
    \]
    where $\mathcal{F}$ is the set of all feasible allocations of $I$. If the
    instance $I$ is obvious from context, we omit it and write $\mu_i$.
\end{definition}

\noindent
A feasible allocation $A$ with $\min_{A_j \in A} v_i(A_j) = \mu_i^I$ for agent
$i$ is said to be an \emph{MMS partition} of $I$ for $i$. A feasible allocation
where each agent $i$ receives a bundle it values at no less than $\mu_i$, is
called an \emph{MMS allocation}. Even without conflicts, there are instances for
which no MMS allocation exists \cite{Kurokawa:2016}. Additionally, as
calculating the MMS of an agent is NP-hard \cite{Woeginger:1997}, finding MMS
allocations is generally infeasible. Instead, approximations are usually
considered. We say that a feasible allocation is $\alpha$-\emph{approximate MMS}
if each agent receives a bundle they value at no less than $\alpha\mu_i$.

Several useful properties of MMS have been found in the ordinary, conflict-free
setting. Many of these are not easily extendable or applicable to allocation
under item conflicts, as will be discussed later; however, the basic properties
of \emph{scale-freeness} and \emph{normalization} may be quite naturally
extended to the new scenario. These properties simplify working with
approximations of MMS, especially when finding polynomial-time algorithms. The
proofs of the properties carry over from unconstrained allocation, and have
been omitted. See, e.g., the works of \citet{Amanatidis:2017} and
\citet{Garg:2019} for details.

\begin{lemma}[Scale-freeness]\label{lem:scale-freeness}
    For a problem instance $I = (N, M, V, G)$, let $I' = (N, M, V', G)$ be the
    problem instance where the valuations of each agent $i$ are scaled by some
    constant $c_i > 0$. Then $\smash{\mu_i^{I'} = c_i\mu_i^I}$ and all
    MMS allocations and MMS partitions of $I$ are also MMS allocations and
    MMS partitions of $I'$.
\end{lemma}

\begin{lemma}[Normalization]\label{lem:normalization}
    If $v_i(M) = |N|$ for an agent $i$ in a problem instance $I=(N, M, V, G)$,
    then $\mu_i^I \le 1$.
\end{lemma}

\noindent
Fair allocation of conflicting items is a generalization of fair allocation
without conflicts, which is the special case of $G = (M, \emptyset)$. A similar
problem to fair allocation of conflicting items, is fair allocation under
\emph{cardinality constraints}~\cite{Biswas:2018}. In this version of the
problem, the items are divided into categories and a bundle may not contain more
items from a single category than some given threshold. Instances of this
version of the problem where each category has a threshold of~1 may be reduced
directly to fair allocation of conflicting items, with the conflict graph
becoming a collection of cliques, one per category.

While we only consider additive valuation functions in our instances, there
exists research on other classes of valuation functions. Some of our results
rely on reduction to unconstrained fair allocation with more complex valuation
functions in order to maintain the conflicts to a certain degree. A function
$f$ is \emph{fractionally subadditive} (XOS) if there exists a finite set $F$ of
additive functions such that for any set $S$, $f(S) = \max_{f' \in F}
f'(S)$. Submodularity is a more restricted case. For a \emph{submodular
function} $f$ and any two sets $S$ and $S'$, $f(S) + f(S') \ge f(S \cup S') +
f(S \cap S')$.

A graph $G = (M, E)$ is \emph{complete} if all vertices are neighbors, and
\emph{empty} if $E=\emptyset$. Given a set of vertices $S \subseteq M$, we let
$G[S]$ denote the \emph{induced subgraph of $S$ on $G$}, i.e., the graph
consisting of the vertices in $S$ and all edges in $E$ that connect pairs of
vertices in $S$. We let $\Delta(G)$ denote the maximum degree of the graph $G$,
and $\CC(G)$ the cardinality of its largest connected component. A
$k$\emph{-coloring} of $G$ is a coloring of the vertices in $G$ using $k$
distinct colors, such that no two neighboring vertices share a color. We let
$\chi(G)$ denote the smallest integer $k$ for which $G$ has a $k$-coloring.
Note that the problem of determining if a graph is $k$-colorable is
NP-complete. However, a $(\Delta(G) +
1)$-coloring always exists and can greedily be found in polynomial time.

\section{Conflicting Items in Theory}
\label{sec:theory}

We are interested in determining, theoretically, to what extent we can guarantee
the agents either fulfillment or approximation of various fairness criteria when
there are conflicting items. Besides theoretical results, we are interested in
the degree to which such fairness can be achieved in polynomial time. In this
section, we explore the existence and non-existence of EF1, both by itself and
as a product of MNW. We also explore approximations to both proportionality and
MMS.

\subsection{Envy-Freeness up to One Good}
\label{sec:envyfreeness}

Without item conflicts, we know that EF1 allocations always exist
\cite{Lipton:2004}. With item conflicts, this is not always the case. The
simplest example is when there are no feasible allocations at all, let alone EF1
allocations. We are more interested in cases where the items may, in fact, be
allocated---but even then there are instances that do not admit an EF1
allocation. The following \lcnamecref{prop:noef1} shows that if an item is in
conflict with as many items as there are agents, there always exists a set of
binary valuation functions that precludes EF1, even if feasible allocations
exist.

\begin{proposition}\label{prop:noef1}
    For any graph $G = (M, E)$ with $\Delta(G)\geq n$, there is a problem
    instance $([n], M, V, G)$ that has no EF1 allocation.
\end{proposition}
\begin{proof}
    First of all, the instance may be unfeasible, i.e., it may have no feasible
    allocation. This would be the case, for example, for the complete graph
    $K_{n+1}$. Assume that there is some feasible allocation, and select some item
    $j$ with degree at least $n$. For every agent, let the neighbors of $j$ get
    a value of~$1$, and let all other nodes get a value of $0$. Some agent $i$
    must receive item $j$, and some agent $i'$ must receive at least two of its
    neighbors, and $i$ will still envy $i'$ after removing one of the items
    allocated to $i'$, which means that no allocation for these valuations can
    be EF1.
\qed
\end{proof}

\noindent
In other words, an EF1 allocation is not guaranteed when the number of agents is
no higher than the maximum degree of the graph. This is rather unsurprising, as
a high
degree in relation to the number of agents limits the number of feasible
allocations quite drastically. In the other direction, it is possible to show
that for graphs with sufficiently small components, EF1 allocations always exist.

\begin{proposition}\label{prop:ef1}
    For any graph $G = (M, E)$ with $\CC(G)\leq n$, all problem instances $([n],
    M, V, G)$ have EF1 allocations that can be found in polynomial time.
\end{proposition}
\begin{proof}
    Let $I$ be the original instance, with
    conflicting items. We construct an instance $I'$ of the fair allocation
    problem with cardinality
    constraints, by introducing one category $C_h$ for each connected component,
    consisting of its vertices, and setting the corresponding threshold to~$1$,
    i.e., no agent can get two items from the same category/connected component.
    Then $I$ is a relaxation of $I'$, and any feasible allocation for $I'$ is
    feasible for $I$. \Citeauthor{Biswas:2018} showed that there exists a
    polynomial-time algorithm that for any instance of the fair
    allocation problem with cardinality
    constraints and additive valuation functions finds an EF1
    allocation~\citep{Biswas:2018}. Using their algorithm, an EF1 allocation can
    be found for $I'$, and consequently one can be found for $I$ as well.
\qed
\end{proof}

\noindent
Propositions \ref{prop:noef1} and \ref{prop:ef1} establish the
existence and non-existence of EF1 allocations at opposite sides of a spectrum,
in a sense, and leave open the question of whether EF1 exists when $\Delta(G) <
n < \CC(G)$. At least for some small subset of such instances, it can be shown
that no EF1 allocations exist, as illustrated by the following
\lcnamecref{exm:noef1-n>delta}.

\begin{example}\label{exm:noef1-n>delta}
    Let $G = K_{3,3}$ (see below) and let there be a total of 4 agents. Then we
    have $\Delta(G) = 3 < 4 = n < \mathcal{C}(G)$. For all agents, let items
    $1$, $2$ and $3$ have a value of $2$ and items $4$, $5$ and $6$ have a value of $3$.

    
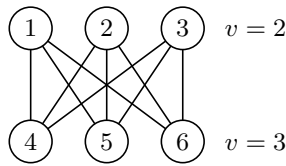
\begin{center}
        \begin{tikzpicture}
            \draw[graph]
                (0, 1.5) node (1) {1}
                (1, 1.5) node (2) {2}
                (2, 1.5) node (3) {3}
                (0, 0) node (4) {4}
                (1, 0) node (5) {5}
                (2, 0) node (6) {6}
                (1) edge (4)
                (1) edge (5)
                (1) edge (6)
                (2) edge (4)
                (2) edge (5)
                (2) edge (6)
                (3) edge (4)
                (3) edge (5)
                (3) edge (6)
                ;
            \draw[overlay]
                (3.east) node[right=4pt] {$v=2$}
                (6.east) node[right=4pt] {$v=3$}
                ;
        \end{tikzpicture}
    \end{center}
    Either
    \begin{inl}
    \item one agent receives a bundle with two or more items worth $3$, or
    \item three agents receive a bundle with one item worth $3$.
    \end{inl}
    In the first case, we cannot guarantee the worst-off agent a bundle worth
    more than $2$. Since there is a bundle with at least two items worth $3$,
    this does not allow for EF1. In the latter case, the last agent receives all
    items worth $2$. After removing any item from this bundle, the value remains
    $4$, which is more than any other agent receives. Consequently, an EF1
    allocation does not exist.
\end{example}

\begin{remark}\label{remark:minimum}
Interestingly, it can be shown that the instance in \cref{exm:noef1-n>delta}
contains the fewest number of items for which there is no EF1 allocation when $n
> \Delta(G)$. For $n = 2$, this follows directly from \cref{prop:ef1}. When $n
\ge 3$ an EF1 allocation must always exist in this situation when $m \le n + 1$,
as each agent can in turn choose their most-valuable remaining item until either
no more remain or each agent has received one item. If $m = n + 1$, the last
item can be given to any agent without a conflicting item. Specifically for
three agents, one can show, by working through the possible cases, that all such
instances with five items admit an EF1 allocation.
\end{remark}

\noindent
The instance in \cref{exm:noef1-n>delta} is not an isolated case. We can in
fact find similar instances for any number of agents greater than three.

\begin{proposition}\label{prop:non-existence-ef1}
    For any $n \ge 4$, there exists a graph $G$ with $n > \Delta(G)$ and a set
    of valuations $V$, so that the instance $([n], M, V, G)$ does not admit any
    EF1 allocations.
\end{proposition}

\begin{proof}
    Let the graph $G$ be the complete bipartite graph $K_{n - 1, n - 1}$. This
    graph contains two partite sets of $n - 1$ vertices, where there are no
    edges between vertices in the same set. For each pair of vertices in
    different sets, there is an edge between them. This results in the graph
    being regular, with a degree of $n - 1$, as each vertex is connected to $n -
    1$ other vertices. If we use this graph as a conflict graph, each agent may
    only receive items from the same partite set.

    For all agents, let the items in one of the sets each have a value of $2$,
    and the items in the other set each have a value of $2n - 5$. Any feasible
    allocation must contain either
    \begin{inl}
    \item a bundle with two or more items worth $2n - 5$ or
    \item $n - 1$ bundles with a single item worth $2n - 5$.
    \end{inl}
    In the first case, all other bundles must contain a value of at least $2n -
    5$ for EF1 to be possible (more if a bundle has $3$ or more items worth $2n
    - 5$). Each bundle must therefore contain at least one item, and at least two
    bundles must contain items only from the set where each item is valued $2$. The
    total value of the items worth $2$ is $2(n - 1)$, allowing the worst-off
    agent to receive a value of at most $n - 1$. When $n > 4$, $2n - 5 > n - 1$
    and this case does not admit an EF1 allocation. For $n = 4$, there are three
    items worth $2$, and thus one of the two agents cannot receive a value of
    more than $2 < 2n - 5$.

    In the second case, the last agent must receive all items worth $2$ and has
    a bundle worth $2n - 2$. Removing any item from the bundle, results in a
    bundle valued at $2n - 4 > 2n - 5$. This means that all other agents envy
    the bundle, even after removing an item. Consequently, an EF1 allocation
    cannot exist.
\qed
\end{proof}

\noindent
\Cref{remark:minimum,prop:non-existence-ef1} show that when $\Delta(G)<n<\CC(G)$,
for $n\geq 4$, there are some instances for which EF1 exists, and some for which
it does not. Giving a more detailed classification or characterization of such
instances remains an open problem.

The instances used in \cref{exm:noef1-n>delta,prop:non-existence-ef1}
rely on very specific valuations for the items in order to show the
non-existence of EF1. Changing the valuations slightly may easily result in the
existence of EF1 allocations. Limiting the possible values each item can take
allows for some special cases where we have existence proofs and polynomial-time
algorithms.

\begin{proposition}
    \label{prop:pathef1}
    If a problem instance $(N, M, V, G)$, where $|N| > 2$, has valuation
    functions $v_i : M \rightarrow \{0, 1\}$, and the components of $G$ are
    paths, then the instance has an EF1 allocation, which may be found in
    polynomial time.
\end{proposition}

\begin{proof}
    If there exists an item $j$ with $v_{ij} = 0$ for each agent $i \in N$, who
    receives $j$ does not affect if an allocation is EF1, except to the
    extent that it may limit possible feasible allocations. However, because
    $|N| > 2 \ge \Delta(G)$, for any feasible partial allocation of $M \setminus
    \{j\}$, there exists at least one agent with no conflicting items to
    $j$. In practice, this means that all items of this type can initially be
    removed and then, later, allocated arbitrarily. Thus, we only need to
    consider instances where there are no such items.

    In order to show that an EF1 allocation exists for all instances where no
    item is valued at $0$ by all agents, we will create a polynomial-time algorithm that
    iteratively allocates items by considering the connected components of $G$
    one by one in arbitrary order. For each component $C$, the items are
    iteratively allocated in order---from one end of the path to the other---while
    maintaining the following three properties.

    \begin{stmts}[widest=iii, itemsep=0.75ex]
        \item \label{item:prop-1} The allocation is EF1.
        \item \label{item:prop-2} There exists an ordering, $\mathcal{O}$, of
            the agents such that no agent envies earlier agents in $\mathcal{O}$.
        \item \label{item:prop-3} If some, but not all, items in the currently
            considered component, $C$, are unallocated, the agent $i$ that
            received the last allocated item in $C$ does not envy any
            other agent.
    \end{stmts}

    For property \ref{item:prop-2}, we rely on maintain an acyclic \emph{envy
    graph}, i.e., a directed graph of the agents, with an edge from agent $i$ to
    agent $i'$ if $i$ envies the bundle of $i'$. When there are no cycles in the
    envy graph, any topological ordering of the graph orders the agents so that
    no agent envies earlier agents. We can easily find a topological ordering in
    polynomial time for an acyclic graph. \Citet{Lipton:2004} give a
    polynomial-time procedure that decycles an envy graph, without breaking EF1
    (property \ref{item:prop-1}), by exchanging bundles along the graph's
    cycles. Note that the procedure of \citeauthor{Lipton:2004} does not swap the
    bundles of agents that prior to the application of the decycling procedure did not
    envy any other agent. Since agent $i$ from \ref{item:prop-3} is an agent of
    this type, the decycling procedure does not change $i$'s bundle. As the
    content of the other bundles does not change, only their owners, agent $i$
    will remain unenvious of all other agents also after applying the decycling
    procedure. Consequently, if \ref{item:prop-1} and \ref{item:prop-3} hold
    prior to application of the decycling procedure, they also hold afterwards.

    Our algorithm starts with an initially empty allocation, for which all three
    properties hold. Then, in each step let $C$ be the current component,
    $\mathcal{O}$ any ordering that fulfills the requirements in
    \ref{item:prop-2} and $i$ the agent from \ref{item:prop-3}. Also, let $j$ be
    the next unallocated item in $C$ and $j'$ the (unallocated) item following
    $j$, if any exists. (If $j$ is the first item in $C$, then there is no agent
    $i$.) Additionally, let $i'$ be the first agent in $\mathcal{O}$ with $i'
    \neq i$ and $v_{i'}(j) = 1$, if any exists. Since each item is valued at
    $1$ by at least one agent, $i'$ always exists when $i$ does not. Then we
    allocate item $j$ and, possibly, $j'$ by the following rules:

    \begin{enumerate}
        \item \label{item:case-1} If $i'$ exists, then allocate $j$ to $i'$.
        \item \label{item:case-2} If $i'$ and $j'$ do not exist, then allocate
            $j$ to the first agent $i'' \in \mathcal{O}$, $i'' \neq i$.
        \item \label{item:case-3} If $i'$ does not exist and $j'$ exists, then
            allocate $j'$ to the first agent $i''' \in \mathcal{O}$ with
            $v_{i'''}(j') = 1$ and $j$ to any $i'' \in N$ with $i'' \neq i'''$
            and $i'' \neq i$.\footnote{An agent $i'''$ always exists, as each item is
            valued at $1$ by at least one agent.}
    \end{enumerate}

    We must now show that in each of these three cases, if \ref{item:prop-1},
    \ref{item:prop-2} and \ref{item:prop-3} hold prior to the allocation, then
    they hold afterwards as well. As seen earlier, property \ref{item:prop-2}
    can be achieved by simply using the procedure of \citeauthor{Lipton:2004} after
    the allocations. Thus, we need only show that after applying the rules,
    but before using the decycling procedure of \citeauthor{Lipton:2004},
    \ref{item:prop-1} and \ref{item:prop-3} hold. Let $A = \langle A_1, A_2,
    \dots, A_{|N|} \rangle$ be the allocation prior to giving away $j$ and,
    possibly, $j'$ in cases \ref{item:case-1}, \ref{item:case-2} and
    \ref{item:case-3}, and $A' = \langle A'_1, A'_2, \dots, A'_{|N|} \rangle$
    the allocation afterwards.

    \paragraph{Case \ref{item:case-1}:}

    The only bundle that changes is $i'$'s bundle. That is $A'_{i'} = A_{i'}
    \cup \{j\}$. For any agent $i^* \in N$, $i^* \neq i'$ we know that either
    $v_{i^*}(A_{i^*}) \ge v_{i^*}(A_{i'})$ or $v_{i^*}(j) = 0$. Otherwise, $i^*$
    would have appeared before $i'$ in $\mathcal{O}$ and been chosen instead of
    $i'$. For any $i^*$ with $v_{i^*}(A_{i^*}) \ge v_{i^*}(A_{i'})$ we have
    $v_{i^*}(A'_{i^*}) \ge v_{i^*}(A_{i'}) = v_{i^*}(A'_{i'}) - v_{i^*}(j)$.
    Thus, since $v_{i^*}(A'_{i'}) = v_{i^*}(A_{i'})$ for all other $i^*$, $A'$
    is EF1 and \ref{item:prop-1} holds for $A'$. Additionally, with binary
    valuations each item is valued at either $0$ or $1$ and combined with EF1
    this implies that $v_{i'}(A_{i'}) \ge v_{i'}(A_{i^*}) + 1$ for all $i^* \in
    N$, $i^* \neq i'$. As a result, we have
    \[
        v_{i'}(A'_{i'}) = v_{i'}(A_{i'}) + v_{i'j} \ge v_{i'}(A_{i^*}) =
        v_{i'}(A'_{i^*})\,,
    \]
    and \ref{item:prop-3} holds for $A'$.

    \paragraph{Case \ref{item:case-2}:}

    The only bundle that changes is $i''$'s bundle. That is $A'_{i''} = A_{i''}
    \cup \{j\}$. Since the only agent that values $j$ at $1$ is $i$, the only
    way that EF1 can be broken in $A'$ is if $i$ is envious of $i''$. However,
    since \ref{item:prop-3} holds for $A$, $v_i(A_{i}) \ge v_i(A_{i''})$ and we
    have $v_{i}(A'_{i}) = v_{i}(A_{i}) \ge v_{i}(A_{i''}) = v_{i}(A'_{i''}) -
    v_{ij}$. Hence $A'$ is EF1 and \ref{item:prop-1} holds for $A'$.  Since
    $j$ is the last item in $C$, the current component does not have any
    unallocated items and the next component, if any, has no allocated items.
    Thus, \ref{item:prop-3} holds for $A'$.

    \paragraph{Case \ref{item:case-3}:}

    The bundles of $i''$ and $i'''$ both change. That is $A'_{i''} = A_{i''}
    \cup \{j\}$ and $A'_{i'''} = A_{i'''} \cup \{j'\}$. Since $i$ is the only
    agent that values $j$ at $1$, the changes to $i''$'s bundle cannot, by the
    same logic as in case \ref{item:case-2}, result in $A'$ not being EF1.
    Similarly, the changes to the bundle of $i'''$ are equivalent to the changes
    to $i'$'s bundle in case \ref{item:case-1} (while the item differs, the
    properties of the change are the same). Thus, $A'$ is EF1 and
    \ref{item:prop-1} holds for $A'$. For \ref{item:prop-3} there are two
    possibilities, if $i''' \neq i$, then $v_{i'''}(j) = 0$ and
    $v_{i'''}(A'_{i''}) = v_{i'''}(A_{i''})$. Thus, due to binary valuations and
    $A$ being EF1, we have (as in case \ref{item:case-1})
    \[
        v_{i'''}(A'_{i'''}) = v_{i'''}(A_{i'''}) + v_{i'''}(j') \ge
        v_{i'''}(A_{i^*}) = v_{i'''}(A'_{i^*})\,,
    \]
    for all $i^* \in N$, $i^* \neq i'''$, and \ref{item:prop-3} holds for $A'$
    when $i''' \neq i$. If $i''' = i$, then we have that $v_i(A'_{i''})
    = v_i(A'_{i''}) + 1$. In which case the step $v_{i'''}(A_{i^*}) =
    v_{i'''}(A'_{i^*})$ from the above equation does not hold when $i^* = i''$.
    Since $v_{i}(A_{i}) \ge v_{i}(A_{i''})$ by \ref{item:prop-3}, the following
    equation holds
    \[
        v_i(A'_i) = v_i(A_i) + v_i(j') \ge v_i(A_{i''}) + v_i(j) =
        v_i(A'_{i''})\,,
    \]
    and \ref{item:prop-3} also holds for $A'$ when $i''' = i$.

    \vspace{.4em}

    \noindent
    It can easily be verified that the three cases cover all possible situations
    and that no agent receives two conflicting items. Since all three properties
    hold before and after each step and all items are eventually given away, the
    algorithm produces complete EF1 allocations. It can also easily be verified
    that each of the described steps can be performed in polynomial time in the
    number of agents and items. Since the number of steps is bounded by the number
    of items, the algorithm runs in polynomial time.
\qed
\end{proof}

\noindent
As mentioned in \cref{sec:prelim}, valid allocations for a given conflict graph
$G$ are merely the $n$-colorings of $G$, and quite a lot of work has been done
on certain forms of fairness in the graph coloring domain. In particular,
a coloring is said to be \emph{equitable} if the number of vertices colored by
any two colors differ by at most one. For identical values, this is
equivalent to EF1. \Citet{Lih:2013} provides an overview of many results of
the equitable coloring problem for various graph classes, where the focus is on
minimizing the number of colors. In our setting, however, the number of colors
is given, and we have already seen in \cref{prop:noef1} that we need
$n > \Delta(G)$ for all conflict graphs to have EF1 instances. In this case, we
are guaranteed an equitable coloring, or, equivalently, an EF1 allocation for
identical values.\footnote{Note also that for identical values, an EF1 solution
will be MNW, and if EF1 exists, all MNW solutions are EF1.}

\begin{proposition}
    \label{prop:equit}
    If a problem instance $(N, M, V, G)$, where $|N|>\Delta(G)$, has identical
    values, i.e., $v_{ij}=v_{i'j'}$ for $i,i'\in N, j,j'\in M$, then the
    instance has an EF1 allocation, which may be found in polynomial time.%
    \qed
\end{proposition}

\noindent
In the context of equitable coloring, this is a well-known result, originating
as a conjecture of Erd\H{o}s. See, e.g., the~\citeyear{Kierstead:2008} paper by
\citet{Kierstead:2008} for some of its history, as well as a simplified proof and
corresponding algorithm.

An interesting result for unconstrained allocation is that MNW leads to EF1
\cite{Caragiannis:2019}. With conflicting items, we have seen that there are
instances that do not admit an EF1 allocation, despite there being many feasible
allocations. Consequently, we cannot guarantee that MNW leads to EF1 in this
setting. Even so, it may be useful to determine whether MNW leads to EF1 when
EF1 allocations \emph{do} exist. This is not the case, as can be seen in the
following \lcnamecref{prop:mnwnotef1}.

\begin{proposition}
    If the graph $G = (M, E)$ is complete or empty, then for any problem
    instance $([n], M, V, G)$ with $n > \Delta(G)$, all MNW allocations are EF1.
    If $G$ is neither complete nor empty, there is a problem
    instance $([n], M, V, G)$ with $n \ge \Delta(G)$ for which there exists at
    least one EF1 allocation, but for which no MNW allocation is EF1.%
    \label{prop:mnwnotef1}%
\end{proposition}

\begin{proof}
    First of all, if $G = K_k$ for some $k \le n$, each bundle can only contain
    at most one item. Obviously, all allocations must then be EF1. Note
    that the case of $k > n$ has no feasible allocations. When $E = \emptyset$,
    the problem is equivalent to ordinary fair allocation (without conflicts).
    For this problem, \citeauthor{Caragiannis:2019} showed that all MNW
    allocations are EF1~\citep{Caragiannis:2019}.

    If neither $G = K_k$ nor $E = \emptyset$, then $G$ contains at least one
    subset of three vertices that form an induced subgraph of either $P_3$ or
    $\overline{P_3}$ (see \cref{fig:mnw-ef1-graph-classes}). If $G = P_3$ or $G
    = \overline{P_3}$, with two agents, the valuations in
    \cref{fig:mnw-ef1-valuations} cause the maximum Nash welfare to be $5$.
    However, the only feasible allocation with a Nash welfare of $5$ is $\langle
    \{1, 3\}, \{2\} \rangle$.  This allocation is obviously not EF1, as agent
    $2$ envies agent $1$, even after removing one item. EF1 allocations do exist
    for these instances, with allocations $\langle \{2\}, \{1, 3\} \rangle$ and
    $\langle \{1\}, \{2, 3\} \rangle$ being EF1 for $P_3$ and $\overline{P_3}$,
    respectively.

    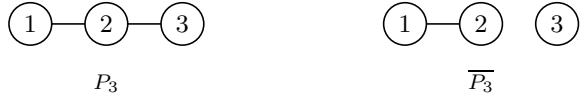
\begin{figure}%
        \hfill
        \begin{tikzpicture}
            \draw[graph, font=\normalsize]
                (0, 0) node (P3-1) {1}
                (1, 0) node (P3-2) {2}
                (2, 0) node (P3-3) {3}
                (P3-1) edge (P3-2)
                (P3-2) edge (P3-3)
                ;

            \node at (1, -.75) {$\vphantom{\overline{P_3}}P_3$};
            \path (0,.75) (0,-1.2); 
        \end{tikzpicture}%
        \hfill
        \begin{tikzpicture}
            \draw[graph, font=\normalsize]
                (4.5, 0) node (P3C-1) {1}
                (5.5, 0) node (P3C-2) {2}
                (6.5, 0) node (P3C-3) {3}
                (P3C-1) edge (P3C-2)
                ;

            \node at (5.5, -.75) {$\overline{P_3}$};

            \path (5,.75) (5,-1.2); 
        \end{tikzpicture}%
        \hfill\mbox{}%
        \caption{The graph classes $P_3$ and $\overline{P_3}$}
        \label{fig:mnw-ef1-graph-classes}
    \end{figure}

    \begin{figure}
        \hfill
        \begin{tabular}{*4c}
            \toprule
            $i$ & $v_{i1}$ & $v_{i2}$ & $v_{i3}$ \\
            \midrule
            1 & 2 & 2 & 3 \\
            2 & 6 & 5 & 6 \\
            \bottomrule
        \end{tabular}%
        \hfill
        \begin{tabular}{*4c}
            \toprule
            $i$ & $v_{i1}$ & $v_{i2}$ & $v_{i3}$ \\ \midrule
            1 & 2 & 1 & 3 \\
            2 & 6 & 5 & 6 \\
            \bottomrule
        \end{tabular}%
        \hfill\mbox{}%
        \caption{Valuations for two agents, where MNW is not EF1, for $G = P_3$
        (left) and $G = \overline{P_3}$}
        \label{fig:mnw-ef1-valuations}
    \end{figure}

    For any graph with both $G \neq K_k$ and $E \neq \emptyset$, we can use the
    valuations for two agents on $P_3$ and $\overline{P_3}$ to construct an
    instance where MNW does not result in EF1. First, let $n = \Delta(G)$ if
    $\Delta(G) + 1 = |M|$ and $n = \Delta(G) + 1$ otherwise. Choose two agents
    and a set $S$ of three items, where the induced subgraph $G[S]$ is either
    $P_3$ or $\overline{P_3}$. Let the two agents value the items in $S$ as
    shown in \cref{fig:mnw-ef1-valuations}, and let them value all other items
    at $0$. Since $G$ contains at least $\Delta(G) + 1$ vertices and $|S| = 3$,
    either $n = \Delta(G) = |M| - 1$ and $|M \setminus S| = \Delta(G) + 1 - 3 =
    n - 2$ or $n = \Delta(G) + 1 \le |M| - 1$ and $|M \setminus S| \ge \Delta(G)
    + 2 - 3 = n - 2$. Consequently, there is always at least $n - 2$ items in $M
    \setminus S$. Let $S'$ be a set of $n - 2$ items from $M \setminus S$ and
    let each of the $n - 2$ remaining agents assign a non-zero value to each
    item in $S'$ and a value of $0$ to the items in $M \setminus S'$.

    The MNW of the constructed instance is non-zero, as it is possible to
    allocate at least one item with non-zero value to each agent. For example,
    the items in $S$ can be allocated to the two agents that assign them a
    non-zero value, in a way that respects the item conflicts and gives each
    agent at least one item from $S$. The remaining $n - 2$ agents can then be
    allocated one item each from $S'$. When $n = \Delta(G)$, this is a feasible
    complete allocation, since $M = S \cup S'$. When $n = \Delta(G) + 1$, the
    instance can contain additional items that are all worth zero to all agents,
    i.e., $M \setminus (S \cup S') \neq \emptyset$. Since each item has at most
    $\Delta(G)$ conflicts, there exists for each item at least one agent that
    has not been allocated a conflicting item, no matter how the other items are
    allocated. Thus, any partial allocation that allocates all goods in $S$ and
    $S'$ can trivially be extended to a complete allocation. Consequently, the
    MNW is only dependent on how $S$ and $S'$ are allocated. Any allocation must
    allocate one item from $S'$ to each of the $n - 2$ agents in order to have a
    non-zero Nash welfare. Each item in $S$ must be given to one of the two
    agents that give them a non-zero value; otherwise, we could increase the
    Nash welfare by doing so. Thus, any MNW allocation must maximize the Nash
    welfare for the two agents on the induced subgraph $G[S]$, and,
    consequently, the allocation cannot be EF1. In the same fashion, a non-MNW,
    EF1 allocation may be constructed by allocating one item in $S'$ to each of
    the $n - 2$ agents and allocating the items in $S$ to the two agents in a
    way that neither envies the other after removing the other's most valuable
    item.%
\qed
\end{proof}

\noindent
\Cref{prop:mnwnotef1} has implications for fair allocation under cardinality
constraints. \Citeauthor{Biswas:2018} have already pointed out that there exist
instances of this problem where MNW does not lead to EF1~\citep{Biswas:2018}.
However, they do not provide any classification of such cases. By using similar
valuations as those used to prove \cref{prop:mnwnotef1}, one can construct
examples whenever there are at least two categories, one of which contains two
or more items and has a threshold of~1. Then the same valuation function
construction as for $\overline{P_3}$ can be used, picking two items from the
category with a threshold of $1$ and any item from another category.

While \cref{prop:mnwnotef1} shows that MNW does not lead to EF1 in general, it
may still be possible that for some combinations of graphs, number of agents and
restricted valuation functions, MNW does lead to EF1. \citet{Biswas:2018} showed
that for cardinality constraints (and other related matroid constraints), MNW
leads to EF1 if all agents have identical valuations. As fair allocation of
conflicting items is equivalent to cardinality constraints when all components are
cliques, MNW leads to EF1 for instances with this class of conflict graph and
identical valuations.\footnote{Note that this argument does not work for all
the instances that were reduced to cardinality constraints in \cref{prop:ef1}.
When the components of the conflict graph are not cliques, there may be
additional MNW allocations for the item conflict instance that are \emph{not}
EF1; indeed, an allocation that is MNW under cardinality constraints might
conceivably, though feasible, not even \emph{be} MNW under item conflicts.}

\subsection{Proportionality}
\label{sec:prop}

While equitable coloring is equivalent to EF1 with identical values, the
\emph{weighted} version has \emph{not} traditionally been a more general EF1
coloring, but is rather linked to the proportion any color receives of the total
weight. That is, each node is weighted, and a weighted \emph{almost-equitable}
coloring is where each of the $n$ colors covers nodes with a total weight of at
least $W/n$, where $W$ is the total of all the weights in the graph. This, of
course, is equivalent to \emph{proportional} allocation, where the agents have
identical valuations (though the items do not necessarily have identical
values). Some results apply even for general valuations (cf.\@
\cref{prop:random}).

\begin{algorithm}
    \label{alg:random}
    This is the randomized coloring procedure with symmetry breaking, described by
    \citet{Pemmaraju:2008}. Specifically, it is the version that applies to
    weighted almost-equitable colorings.
    \begin{enumerate}
    \item Let $\pi$ be a permutation of $M$, selected uniformly at random.
    \item Tentatively allocate each item to an agent, selected uniformly at
        random.
    \label{ln:randpick}%
    \item If an agent received conflicting items $j$ and $j'$, where $\pi_j<\pi_{j'}$,
    deallocate $j'$.
    \label{ln:partial}%
    \item Allocate remaining items arbitrarily, avoiding conflicts.
    \label{ln:arbitrary}%
    \end{enumerate}
\end{algorithm}

\noindent
Note that in step~\ref{ln:partial}, conflicts are based on the initial,
tentative allocation, before any deallocations have taken place. Items $j$ and
$j'$ are considered to be in conflict even if $j$ has already been deallocated,
so the order of deallocation is irrelevant.

It should be obvious that this algorithm runs in polynomial time. What is more,
it provides each agent with an expected fraction of $1 - 1/e$, or almost 2/3, of
a proportional share.

\begin{proposition}
    Given a problem instance $(N, M, V, G)$, where $|N|>\Delta(G)$,
    \Cref{alg:random} finds an approximate proportional allocation with item
    conflicts, where each agent can expect to get at least a fraction of $1 -
    1/e$ of its share.%
    \label{prop:random}
\end{proposition}
\begin{proof}
    Adapted from Lemma~5 of \citet{Pemmaraju:2008}. Specifically, their
    result uses exactly $\Delta(G)+1$ colors, whereas we need to permit any
    $n>\Delta(G)$. We wish to describe $p_{ij}$, the probability that item $j$
    is allocated $i$ after step~\ref{ln:partial}. (Note that this value will be
    identical for all $i$.) This will tell us the expected value of any bundle
    at that point, a value that can only increase in step~\ref{ln:arbitrary}.
    That is, if we can show that $p_{ij}\geq(1-1/e)/n$ for every item $j$, each
    bundle will---in expectation---be valued at no less than $(1-1/e)\cdot
    v_i(M)/n$ by agent $i$. In other words, each agent can expect a fraction of
    $1-1/e$ of its share.

    The probability that $j$ is allocated to $i$ in step~\ref{ln:randpick} is
    $1/n$. If exactly $k$ of its neighbors are ranked earlier by $\pi$, the
    probability that $j$ is not removed again in step~\ref{ln:partial} is
    $(1-1/n)^k$, meaning that we would have $p_{ij}=1/n\,\cdot\,(1-1/n)^k$. The
    probability of exactly $k$ neighbors are ranked before $j$ is
    $1/(\delta(j)+1)$, where $\delta(j)$ is the degree of node $j$ in $G$, which
    gives us:
    \begin{IEEEeqnarray*}{rCl}
        p_{ij}
    &=&
        \frac{1}{n}\sum_{k=0}^{\delta(j)}\frac{1}{\delta(j)+1}
        \cdot
        \biggl(1-\frac{1}{n}\biggr)^k
    =
        \frac{1}{\delta(v)+1}
        \biggl(1 -
            \biggl(1 -
                \frac{1}{n}
            \biggr)^{\delta(v)+1}
        \biggr)
    \\
    &\geq&
        \frac{1}{n}
        \biggl(1 -
            \biggl(1 -
                \frac{1}{n}
            \biggr)^{n}
        \biggr)
    \geq
        \frac{1}{n}
        \biggl(
            1 - \frac{1}{e}
        \biggr)
    \end{IEEEeqnarray*}
    This, then, produces the desired bound on the individual expectations.
\qed
\end{proof}

\noindent
Approximate proportionality implies approximate MMS, that is, if you receive a
fraction $\alpha$ of your proportional share, you receive \emph{at least} a
fraction $\alpha$ of your maximin share. Note, however, that
\cref{prop:random} only bounds the \emph{individual} expectations. That is,
this is not the expectation of the proportionality or MMS \emph{guarantee},
i.e., the lowest fraction received by any agent.

The following result provides a deterministic bound for such a guarantee for
proportionality, in a scenario where all agents have identical valuation
functions.\footnote{In this case, MMS allocations are guaranteed to exist.}

\begin{proposition}
    For a problem instance $(N, M, V, G)$, where $|N|>\Delta(G)$, if agents have
    identical valuations $v:M\to[0,1]$, there exists an $\alpha$-approximate
    proportional allocation with
    \[
    \alpha\geq 1 - 1/e - cn\sqrt{\ln(n)/W}\eqcomma
    \]
    for any $c>\sqrt{8}$, where $W=v(M)$.%
    \label{prop:randprop}
\end{proposition}
\begin{proof}
    Based on Lemma~6 of \citet{Pemmaraju:2008}, adapted from using exactly
    $\Delta(G)+1$ colors to permitting any $n>\Delta(G)$. Let $S$ be the bundle
    assigned to some agent by \cref{alg:random}. The strategy is to show that
    for any fixed $c>0$,
    \begin{equation}\label{eq:prbound}
        \Pr\biggl[
            \,
            v(S) <
            \biggl(1-\frac{1}{e}\biggr)\cdot\frac{W}{n}-c\sqrt{W\ln n}
            \,
        \biggr]
        \leq
        \frac{1}{n^{c^2/8}}
        \eqdot
    \end{equation}
    The probability that the leftmost inequality holds for \emph{at least one}
    of the bundles is then, by the union bound, at most $n/n^{c^2/8} =
    n^{1-c^2/8}$. Choosing $c>\sqrt{8}$ leads to a positive probability that
    \emph{no} agents receives less than $(1-1/e)W/n - c\sqrt{W\ln n}$, which
    means that such an allocation must exist. Dividing by $W/n$ yields the
    approximation ratio $\alpha$.

    For simplicity, we renumber the items according to $\pi$,
    and define $j\prec k$ to mean that $j<k$ and $jk\in E$.
    Let $a:M\to N$ represent the tentative allocation after
    step~\ref{ln:randpick}. We now consider differences in conditional
    expectations along $\pi$, for use in Azuma's martingale inequality:
    \[
        c_j = \bigl|
        \Ex[v(S) \mid a_1 = i_1,\dots,a_j=i_j]
        -
        \Ex[v(S) \mid a_1 = i_1,\dots,a_j=i'_j]
        \bigr|
    \]
    There may be no difference (i.e., we may have $c_j=0$), as the change need
    not affect our bundle $S$. If, however, item $j$ is added or removed,
    an expected proportion of $1/n$ of its lower-priority conflicting items will
    consequently be removed or added, respectively. Thus we can bound the
    differences by:%
    \footnote{\Citeauthor{Pemmaraju:2008} have
    a factor of $2$ rather than $1$ in front of the second term, leading to
    $c>\sqrt{18}$, rather than $c>\sqrt{8}$~\citep{Pemmaraju:2008}. This is
    presumably to account for a fraction of $1/n$ nodes lost as well as
    gained, leading to an overly conservative bound.}%
    \[
        c_j \leq
        \biggl|
        \mkern2mu
        v_j - \frac{1}{n}\sum_{\mathclap{j\prec k}} v_k
        \mkern2mu
        \biggr| \leq
        v_j + \frac{1}{n}\sum_{\mathclap{j\prec k}} v_k
    \]
    Because values are in $[0,1]$, we have $c_j\leq 1 + \Delta(G)/n < 2$. We
    also have
    \[
        \sum_{j=1}^m c_j
        \leq
        \sum_{j=1}^m v_j +
        \frac{1}{n}\sum_{j=1}^m
            \sum_{
            j\prec k
        } v_k \leq
        W + \frac{\Delta(G)}{n} \cdot \sum_{j=1}^m
        v_k
        <
        2W
        \eqcomma
    \]
    and $\sum_{j=1}^m c^2_j < 2\sum_{j=1}^m c_j < 4W$. Azuma's inequality then
    produces
    \[
        \exp\left(\frac{-t^2}{2\sum_{j=1}^mc_j^2}\right)
        \leq
        \exp\left(\frac{-c^2W\ln n}{8W}\right)
        =
        \frac{1}{n^{c^2/8}}
        \eqcomma
    \]
    which gives us \eqref{eq:prbound}.%
\qed
\end{proof}

\noindent
While the bound given by \cref{prop:randprop} may be rather weak for many
practical instances, it does show that, given the assumptions about $V$, for a
given number of agents, the guaranteed fraction of proportionality improves as
the number of items (or, rather, the sum of their values) grows. That is,
\[
    \alpha(W) = 1 - 1/e - o(1)\eqcomma
\]
where $W=v(M)$. (\Citet{Pemmaraju:2008} also present two other similar
guarantees, but they are given without proof, and so are not as straightforward
to generalize to an arbitrary number of agents.)

\subsection{Maximin Shares}
\label{sec:mms}

A common strategy for unconstrained MMS approximation is to use
algorithms that build on the basic concept of bag filling, where items are
placed in a bag one by one. If the
bag reaches a certain value for some remaining agent, it is allocated
as this agent's bundle. Then, the algorithm continues drawing items into a new
bag, repeating the procedure until there are no remaining agents. It can easily
be shown that a $(1 - \alpha)$-MMS approximation can be achieved using this
method when no item is worth more than $\alpha\mu_i$.

Bag filling is usually combined with preprocessing procedures, in order to
reduce the impact of high-valued items on the achievable $\alpha$. A common
procedure is to carefully construct a small set of high-valued items that
combined are worth at least $\alpha\mu_i$ to some agent $i$. If the MMS of each
agent is not reduced in the instance created by removing $i$ along with the
items in the constructed set, i.e., assigning the set as $i$'s bundle, then one
can solve the problem for the reduced instance instead. A simple variant is to
allocate individual items worth at least $\alpha\mu_i$.  Other somewhat typical
steps are conversion to so-called \textit{ordered instances}
\cite{Bouveret:2016} and creating initial bundles of high-valued items for the
bag filling to fill with low-valued items \citep[see, e.g.,][]{Garg:2020}.

For instances with conflicting items, the conflict graph and especially its
degree, $\Delta(G)$, limit the viability of these strategies. Most are simply
not usable as they would cause the creation of infeasible bundles and their
correctness proofs fail. Other strategies can be modified to work with the
conflict graph, albeit at the cost of reduced efficiency. Bag filling is one of
these, where instead of using $M$ as the source of items, we select one of a set
of feasible bundles as the source for each bag to be filled, as described in the
following \lcnamecref{alg:bag-filling}.

\begin{algorithm}\label{alg:bag-filling}
    Takes as input a problem instance $(N, M, V, G)$, a subset $N' \subseteq N$,
    a feasible partition $A = \langle A_1, A_2, \dots, A_\ell \rangle$ of some
    $M' \subseteq M$, where $\ell \ge 1$, and a limit $x_i > 0$ for each $i
    \in N'$.
    Repeating the following steps until no suitable $A_k$ is available will
    allocate bundles to some subset of agents $N'' \subseteq N'$, with each
    bundle worth at least $x_i$ to the agent $i$ that receives it.

    \begin{enumerate}
        \item Let $A_k\in A$ be a bundle with $v_i(A_k) \geq x_i$ for some $i\in
            N'$.
        \item Let $B$ be an empty bag. \label{item:bundle}%
        \item Add items from $A_k$ to $B$, one at a time, until $v_i(B) \ge x_i$
            for some $i \in N'$. \label{item:bag-filling}%
        \item Give $B$ to some $i \in N'$ with $v_i(B) \ge x_i$, and let $N' =
            N' \setminus \{i\}$ and $A_k = A_k \setminus B$.%
            \label{item:allocation}%
    \end{enumerate}
\end{algorithm}

\noindent
Separating the items into several feasible sets prevents the creation of
unfeasible bundles. However, it will cause some items to stagnate, as some sets
may be valued almost, but not quite, $x_i$ and will therefore not be used in the
algorithm. Consequently, the guaranteed value each agent will receive is
generally lower with multiple sets, if all agents are to receive a
bundle. \Cref{lem:bag-filling} provides a general formula to calculate the
possible guarantees, which depends on the value of $M'$, and the number of
source bundles, $\ell$, and agents, $|N'|$.

\begin{lemma}\label{lem:bag-filling}
    Given a problem instance $(N, M, V, G)$, a subset of agents $N' \subseteq N$
    and a feasible partition $A = \langle A_1, A_2, \dots, A_\ell \rangle$ of
    some subset of items $M' \subseteq M$, where $\ell \ge 1$,
    \cref{alg:bag-filling} allocates a feasible bundle to each agent in $N'$ if
    $v_{ij} \le x_i$ and $x_i \le v_i(M')/(\ell + 2(|N'| - 1))$ for all $j \in
    M', i \in N'$. Items that are allocated are taken from at most $\min(\ell,
    |N'|)$ sets in $A$ and the procedure runs in polynomial time.
\end{lemma}

\begin{proof}
    First, note that each allocated bundle consists of items taken from a single
    set $A_k \in A$. Since $A$ is feasible, each allocated bundle must be as
    well. Additionally, since each bundle is created from a single set in $A$,
    the algorithm cannot use items from more than $|N'|$ sets in $A$, or $\ell$
    sets if $|N'| > \ell$. It is obvious that the procedure runs in polynomial
    time, as all operations can easily be verified to be performable in
    polynomial time in the number of agents and items.

    If at all points during the execution of the algorithm, the remaining value
    in $A$, for any remaining agent, $i$, is no less than $\ell x_i$, there is
    always a bundle $A_k \in A$ with $v_i(A_k) \ge x_i$. Consequently,
    step~\ref{item:bag-filling} will at some point create a bundle worth $x_i$
    to $i$, that is allocated to $i$ in step~\ref{item:allocation}.
    Note that since step~\ref{item:bag-filling} adds items to $B$
    one at a time, when $v_{ij} \le x_i$, each remaining agent, $i$, will value
    each already allocated bundle at no more than $2x_i$. Otherwise, the value
    would already be at least $x_i$ before adding the last item. Hence, in the
    worst case, where $|N'| - 1$ agents have been allocated a bundle, the
    remaining agent $i$ will always value the remaining items at no less than
    $\ell x_i$ if
    \begin{align*}
        \ell x_i &\le v_i(M) - (|N'| - 1)2x_i \\
        x_i &\le \frac{v_i(M)}{\ell + 2(|N'| - 1)}\,,
    \end{align*}

    \noindent
    which is the condition in the lemma.
\qed
\end{proof}

\noindent
The condition used in the proof of \cref{lem:bag-filling} where the remaining
value is no less than $\ell x_i$, while sufficient, is overly strict for many
inputs. In practice, one could, e.g., have that all remaining items are located
in the same $A_k \in A$. In this case, it would be sufficient that the remaining
value is no less than $x_i$. However, in the worst case, where the remaining
value is spread evenly across $A$, then $\ell x_i$ is the minimum possible value
for which the algorithm will allocate a bundle to each agent in $N'$.

As with unconstrained bag filling, the guarantees of \cref{alg:bag-filling} are
dependent on the value of the highest-valued items. Thus, a way to handle
high-valued items is needed. The standard method used without item conflicts,
is to create a maximal matching between agents and items they value at least
$\alpha\mu_i$. Allocating items based on this matching guarantees that a subset
of agents each receives a one-item bundle valued at least $\alpha\mu_i$. In the
reduced instance created by removing all matched agents and items, each agent
has at least the same MMS as in the original instance and no items are worth
$\alpha\mu_i$ or more. However, with a conflict graph, there are severe
limitations on which allocations are feasible, and these are highly dependent on
the number of agents. Consequently, a reduction using this type of matching of
agent--item-pairs may lead either to a reduction in maximin shares or to the
reduced instance becoming infeasible. Nonetheless, we can replicate some similar
results when constrained by a conflict graph.

\begin{lemma}\label{lem:alpha-reduction}
    For a problem instance $I = (N, M, V, G)$, with $|N| > \Delta(G)$, let $N'
    \subseteq N$ and $M' \subseteq M$, such that a perfect matching exists
    between the agents in $N'$ and items they value at no less than
    $\alpha\mu_i^I$ in $M'$. Let $I' = (N \setminus N', M \setminus M', V',
    G')$, where $V' = \{v_i:i\in N'\}$ and $G'=G[M\setminus M']$. Then,
    \begin{stmts}[widest=iii, itemsep=0.75ex]
        \item $v_i(M \setminus M') \ge |N \setminus N'|\mu_i^I$ for all $i \in
            N$;
        \item $|N \setminus N'| > \Delta(G') \implies
            \smash{\mu_i^{I'}} \ge \smash{\mu_i^{I}}$ for all $i \in N \setminus
            N'$; and
        \item given a feasible partial allocation of $I'$, where each agent $i
            \in N \setminus N'$ receives at least $\alpha\mu_i^{I}$, a feasible
            complete $\alpha$-approximate MMS allocation of $I$ can be
            constructed in polynomial time.\label{item:partial-to-complete}%
    \end{stmts}
\end{lemma}

\begin{proof}
    Let $A = \langle A_1, A_2, \dots, A_{|N|}\rangle$ be an MMS partition of the
    original instance for one of the agents, $i \in N$. After removing the items
    in $M'$, at least $|N| - |M'| = |N| - |N'| = |N \setminus N'|$ bundles in
    $A$ remain unchanged. Each unchanged bundle is valued at no less than
    $\mu_i^I$ by agent $i$.
    Consequently, $v_i(M \setminus M') \ge |N \setminus N'|\mu_i^I$.
    Additionally, if $|N \setminus N'| > \Delta(G')$, it is possible to chose
    $|N \setminus N'|$ unchanged bundles in $A$ and reallocate the remaining
    items in the other bundles to these. This will be a feasible allocation of
    $I'$, with each bundle worth at least $\mu_i^I$ to $i$, i.e., $\mu_i^{I'}
    \ge \mu_i^I$.

    For \ref{item:partial-to-complete}, any perfect matching of the agents in
    $N'$ to items in $M'$ they value at no less than $\alpha\mu_i^I$ can be used
    to turn the feasible partial allocation of $I'$ into a feasible partial
    $\alpha$-approximate MMS allocation of $I$. The currently unallocated items
    can be arbitrarily allocated to agents that have no conflicting items, to
    create a complete $\alpha$-approximate MMS allocation. Since $|N| >
    \Delta(G)$, at least one such agent exists for each unallocated item. Both
    of these steps can easily be performed in polynomial time.
\qed
\end{proof}

\noindent
A matching of the type required for \cref{lem:alpha-reduction} can easily be
found in the same way as in the unconstrained case. This is done by finding a
maximum-cardinality matching in a bipartite graph of the agents and items,
connecting each agent and item by an edge if the agent values the item at no
less than $\alpha\mu_i$.\footnote{It is not strictly necessary to find a
maximum-cardinality matching---a maximal matching is all that is needed. In
other words, one may arbitrarily allocate items one by one, to agents who value
them at no less than $\alpha\mu_i$.} Thus, finding an $\alpha$-approximate MMS
allocation is equivalent to finding a feasible partial allocation in the reduced
instance, where each agent $i$ receives a bundle valued at least
$\alpha\mu_i^I$. We can use \cref{alg:bag-filling} on varying colorings of $G$
to achieve this, providing different guarantees for $\alpha$ depending on the
conflict graph and the number of agents in the reduced instance.

\begin{lemma}\label{lem:alpha-limits}
    For a problem instance $(N, M, V, G)$, with $|N| > \Delta(G)$, let $n'$ be
    the number of remaining agents after allocating items as described in
    \cref{lem:alpha-reduction}, so that $v_{ij} \le \alpha\mu_i$ for all
    remaining agents and items. Then there exists a feasible
    $\alpha$-approximate MMS allocation if $\alpha$ is at most

    \[
        \begin{cases}
            1 & \text{if } n' \le 1\eqcomma \\
            n'/(3n' - 4) & \text{if } n' = \Delta(G) + 1 = \chi(G)\eqcomma \\
            n'/(2n' + \chi(G) - 3) & \text{if } n' \ge \chi(G)\eqcomma \\
            n'/(3n' - 3) & \text{if } n' < \chi(G)\eqdot \\
        \end{cases}
    \]
\end{lemma}

\begin{proof}
    Throughout the proof, assume that the valuations have been scaled so that
    $\mu_i = 1$ for all agents $i$. As long as it can be shown that there exists
    a partial allocation where each agent in the reduced instance receives a
    bundle valued at least $\alpha$, then by \cref{lem:alpha-reduction} a
    feasible complete $\alpha$-approximate MMS allocation exists in the original
    instance. Note that \cref{lem:alpha-reduction} also guarantees that the
    items in the reduced instance are valued at no less than $n'$ by each agent
    in the reduced instance.

    \begin{case}($n' \le 1$).
        The case of $n' = 0$ is trivial. If $n' = 1$,
        at least one bundle of any of the remaining agent's MMS partitions of the
        original instance remains unchanged after the reduction. Thus, the agent can
        be given this bundle valued at no less than $1$.
    \end{case}

    \begin{case}($n' \ge \chi(G)$).\label{case:n-chi-g}
        Create an MMS partition of the original instance for some agent $i$ in
        the reduced instance. Remove all items from the MMS partition that are
        not in the reduced instance and perform \cref{alg:bag-filling} on the
        MMS partition with all remaining agents, never choosing $i$ when
        breaking a tie in step
        \ref{item:allocation}. Finally, create a $\chi(G)$-coloring of the
        remaining unallocated items in the reduced instance and perform
        \cref{alg:bag-filling} on this coloring using the remaining agents.
        Agent $i$ is guaranteed to receive a bundle valued at $\alpha$ in the
        first bag filling if $\alpha \le 1$, as there are at least $n'$ bundles
        valued at no less than $1$ in the MMS partition. Any agent that did not
        receive a bundle in the first bag filling valued the bundle of $i$ at
        less than $\alpha$. Thus, if there are $n''$ agents remaining after the
        first bag filling, it follows from \cref{lem:bag-filling} that all
        agents will receive a bundle in the second bag filling if:
        \begin{align*}
            \chi(G)\alpha &\le n' - (n' - n'' - 1)2\alpha - \alpha - (n'' -
            1)2\alpha \\
            \chi(G)\alpha &\le n' - (n' - 2)2\alpha - \alpha \\
            \alpha &\le n'/(2n' + \chi(G) - 3)
        \end{align*}
    \end{case}

    \begin{case}($\chi(G) > n'$).
        This case can be solved in a similar way to Case \ref{case:n-chi-g}.
        However, instead of performing a second bag filling with a
        $\chi(G)$-coloring, choose an agent $i'$ that did not receive a bundle
        in the first bag filling. For this agent, create an MMS partition of the
        original instance and remove all items that are currently allocated,
        i.e., let only unallocated items in the reduced instance remain. Perform
        \cref{alg:bag-filling} on this partition using all remaining agents.
        Repeat this last bag filling step, exchanging $i'$ with some remaining
        agent, until there are no remaining agents. As in Case
        \ref{case:n-chi-g}, agent $i$ receives a bundle worth at least $\alpha$,
        when $\alpha \le 1$. For the remaining agents, we require that they
        receive a bundle when \cref{alg:bag-filling} is performed on their
        modified MMS partition.  \cref{lem:alpha-reduction} guarantees that for
        any agent $i'$, there are at least $n'$ bundles in their MMS partition
        with a value of at least $1$ each, before removing the allocated items
        in the reduced instance. Thus, it must be the case that after the earlier
        bag fillings, there remains enough value across these $n'$ bundles that
        $i'$ is guaranteed to receive a bundle in the bag filling on its
        modified MMS partition. By the same logic as in Case \ref{case:n-chi-g},
        this is by \cref{lem:bag-filling} guaranteed if:
        \begin{align*}
            n'\alpha &\le n' - (n' - 2)2\alpha - \alpha \\
            \alpha &\le n'/(3n' - 3)
        \end{align*}
    \end{case}

    \begin{case}($n' = \Delta(G) + 1 = \chi(G)$).
        Create an MMS partition of the reduced instance for some agent $i$ and
        perform \cref{alg:bag-filling} on this partition with all other agents
        in the reduced instance. For these $n' - 1$ agents, we can choose an
        $\alpha = n'/(3n' - 4)$, by \cref{lem:bag-filling}. Since there are $n'$
        bundles in the MMS partition and at most $n' - 1$ were used in the bag
        filling, there is at least one untouched bundle in the MMS partition and
        this can be given to $i$ guaranteeing a value of at least $1$. In this
        specific situation, this method is slightly better than Case
        \ref{case:n-chi-g}. \qed
    \end{case}
\end{proof}

\noindent
The minimum guaranteed value each agent will receive from
\cref{alg:bag-filling}, is, as we have seen, highly dependent on the number of
bundles in the supplied partition of $M'$. The fewer bundles, the less value is
required to guarantee the last agent a bundle and the more value can be given
away to each agent. While it is never possible to partition $G$ into fewer than
$\chi(G)$ feasible bundles, we saw in \cref{lem:alpha-limits} that using specific
partitions of $G$ provided a skewed distribution of the items across the bags,
resulting in fewer wasted items. \Cref{lem:alpha-limits} does, however, not fully
take into consideration the possibilities for reallocating items between bundles in
order to slightly reduce the value of the unallocated items at the end of the
bag filling. Depending on the situation, it may be possible to further limit the
type of partition of $M'$ used, or either combine or reorganize some of the bags
during the execution of \cref{alg:bag-filling}, e.g., if the induced subgraph of
$G$ over the remaining unallocated items can be colored with fewer than
$\chi(G)$ colors. While we have not been able to find any general improvements,
it is possible to use these insights to find slightly better guarantees when
there are three agents.

\begin{lemma}\label{lem:alpha-limit-3}
    For a problem instance $(N, M, V, G)$, with $|N| = 3 > \Delta(G)$, there
    exists a $2/3$-approximate MMS allocation.
\end{lemma}

\begin{proof}
    This result follows partially from \cref{lem:alpha-limits}. The lemma
    provides a guarantee of at least $2/3$ in all cases, except when both $n' =
    3$ and $\chi(G) \in \{2, 3\}$. These cases have the slightly worse guarantee
    of $3/5$ and must be handled differently. For the remainder of the proof, we
    assume that the values have been scaled so that for each agent $i$, $\mu_i =
    1$. Consequently, for each agent $i$, we have that $v_i(M) \ge 3$ and since
    $n' = 3$, all items are worth less than $2/3$.

    Let $A = \langle A_1, A_2, A_3 \rangle$ be an MMS partition for one of the
    agents. Let $i$ and $i'$ be the two other agents. If there is a way to
    create bundles $B_i$ and $B_{i'}$, worth at least $2/3$ to respectively $i$
    and $i'$ such that $(B_i \cup B_{i'}) \subseteq (A_k \cup A_{k'})$, with $k,
    k' \in \{1, 2, 3\}$, then a $2/3$-approximate MMS allocation exists.

    If there are two distinct bundles $A_k, A_{k'} \in A$ such that $v_i(A_k)
    \ge 2/3$ and $v_{i'}(A_{k'}) \ge 2/3$, assigning $A_k$ to $i$ and $A_{k'}$
    to $i'$ is sufficient for a $2/3$-approximate MMS allocation. Otherwise, both
    $i$ and $i'$ value only a single bundle $A_k \in A$ at~$2/3$ or more. This
    bundle is then worth at least $5/3$ to both agents. We will show that it is
    always possible to partition $A_k$ together with possibly some items from
    another bundle $A_{k'} \in A$, into two feasible bundles $B_1$ and $B_2$
    both worth at least $2/3$ to $i$.  Consequently, one of the bundles must be
    worth at least $5/6$ to $i'$ and a $2/3$-approximate MMS allocation exists.

    If $v_i(A_k) \ge 2$ or there is at most one item in $A_k$ valued more than
    $1/3$, then $A_k$ can be partitioned into two bundles worth at least $2/3$
    by performing bag filling, always selecting the most valuable remaining
    item.  Thus, assume that $5/3 < v_i(A_k) < 2$ and that the two most valuable
    items in $A_k$, $j$ and $j'$, are such that $1/3 < v_{ij'} \le v_{ij} <
    2/3$. Let $A_k' = A_k \setminus \{j, j'\}$. Then $v_i(A_k') > 1/3$ and if
    $v_i(A_k') \ge 2/3$, then $\langle A_k', \{j, j'\} \rangle$ is a set of
    bundles that satisfies our requirements.

    We now assume that $1/3 < v_i(A_k') < 2/3$, which implies that $v_{ij} >
    1/2$.  Let $C_j$ and $C_{j'}$ be the set of items in $M$ that conflict with,
    respectively, $j$ and $j'$. If $v_i(C_j) \le 2/3$, there must exist a set
    $A_{k'} \in A \setminus \{A_k\}$, such that $v_i(A_{k'} \setminus C_j) \ge
    1/6$. In this case, $\langle \{j\} \cup (A_{k'} \setminus C_j), A_k
    \setminus \{j\}\rangle$ is a set of bundles that satisfies our requirements.
    If $v_i(C_j) > 2/3$, on the other hand, then because $|C_j| \le \Delta(G)$,
    there must exist an item $j'' \in C_j$, with $v_i(j'') > 1/3$.  Depending on
    whether $j''$ is in $C_{j'}$ or not, one can combine either $A_k'$ or $j'$,
    respectively, with $j''$ to produce a bundle worth at least $2/3$, which can
    be complemented by a bundle of the other items in $A_k$.
\qed
\end{proof}

\noindent
Combining the results of
\cref{lem:alpha-reduction,lem:alpha-limits,lem:alpha-limit-3}, it can be shown
that any instance with $|N| > \Delta(G)$ has an $\alpha$-approximate MMS
allocation with $\alpha > 1/3$.%
\footnote{Post-publication note: The results of \citet{Li:2021} can be used to
show the existence of $11/30$-approximate MMS allocations, providing a better
result for the case $n \ge 4$ and $\chi(G) > 11$.}

\begin{theorem}\label{thr:mms-alpha-existence}
    For a problem instance $(N, M, V, G)$, with $|N| > \Delta(G)$, there exists
    an $\alpha$-approximate MMS allocation, where $\alpha > 1/3$. Specifically,
    \[
        \alpha =
        \begin{cases}
            1 & \text{if } n \le 2 \eqcomma \\
            \frac{2}{3} & \text{if } n = 3 \eqcomma \\
            \frac{n}{2n - 1} & \text{if } n \ge 4 \text{ and } \chi(G) = 2
            \eqcomma \\
            \frac{\chi(G)}{3\chi(G) - 3} & \text{if } n \ge 4 \text{ and }
            \chi(G) \ge 3 \eqcomma \\
        \end{cases}
    \]
    where $n=|N|$.
\end{theorem}

\begin{proof}
    We proceed by cases.

    \begin{case}($n \le 2$).
        The case of $n = 1$ is trivial. For $n = 2$, this follows from
        \cref{lem:alpha-limits}. The same result could also be achieved by using
        the divide and choose protocol~\citep{Steinhaus:1948}, which works in a
        similar way to the method in the lemma for $n = 2$.
    \end{case}

    \begin{case}($n = 3$).
        Follows directly from \cref{lem:alpha-limit-3}.
    \end{case}

    In the remaining two cases, we use the fact that
    \cref{lem:alpha-reduction,lem:alpha-limits} guarantee the existence of
    $\alpha$-approximate MMS allocations if $\alpha$ is not higher than the
    limits in \cref{lem:alpha-limits}. Since we cannot guarantee that a certain
    number of items are valued more than $\alpha$ in the general case, the
    highest value for $\alpha$ we may choose is the one that works for all $n'$.

    \begin{case}($n \ge 4$ and $\chi(G) = 2$).
        Either $n' = 1$ and $\alpha$ can be $1$, or $n' \ge \chi(G)$ and we
        require $\alpha \le n'/(2n' - 1)$. This is minimized when $n' = n$ and
        thus $\alpha = n/(2n - 1)$ will work.
    \end{case}

    \begin{case}($n \ge 4$ and $\chi(G) \ge 3$).
        Here, one can easily verify that the guarantee of
        \cref{lem:alpha-limits} increases on the interval $[\chi(G), n]$ and
        decreases on the interval $[1, \chi(G)]$, if the improvements in the
        special case of $\chi(G) = \Delta(G) + 1$ are ignored. Thus, the worst
        case, $n' = \chi(G)$, allows for $\alpha = \chi(G)/(3\chi(G) - 3)$.
        \qed
    \end{case}
\end{proof}

\noindent
Note that \cref{thr:mms-alpha-existence} does not handle the case of $\chi(G) =
\Delta(G) + 1$ by itself, for which \cref{lem:alpha-limits} provides better
guarantees. \cref{remark:mms-guarantee} shows the achievable guarantees in this
case, which can be proven in a similar way to \cref{thr:mms-alpha-existence}.

\begin{remark}\label{remark:mms-guarantee}
    For a problem instance $(N, M, V, G)$, with $|N| > \Delta(G)$, $|N| \ge 4$
    and $\Delta(G) + 1 = \chi(G)$, slightly better guarantees than in
    \cref{thr:mms-alpha-existence} can be found for $\Delta(G) \ge 2$ using
    \cref{lem:alpha-limits}. Specifically,
    \[
        \alpha =
        \begin{cases}
            \frac{\chi(G) + 1}{3\chi(G) - 1} & \text{if } \chi(G) < 7
            \eqcomma \\[1ex]
            \frac{\chi(G) - 1}{3\chi(G) - 6} & \text{if } \chi(G) \ge 7
            \eqdot \\
        \end{cases}
    \]
\end{remark}

\noindent
\cref{thr:mms-alpha-existence} shows the existence of $\alpha$-approximate MMS
allocations with approximation factors better than $1/3$. The proof of the
theorem is in fact constructive. However, the method used in the proof relies
heavily on being able to both find a minimum vertex coloring of $G$ and being
able to find MMS partitions for the agents. Both problems are NP-hard, leaving
us without a polynomial-time approximation algorithm. Relaxing the approximation
guarantees and employing some of the tricks used in MMS approximation when there
are no restrictions on the bundles, it is possible to find a polynomial-time
approximation algorithm for cases where $|N| > \Delta(G)$.
\cref{alg:mms-polynomial} outlines this algorithm.

\begin{algorithm}\label{alg:mms-polynomial}
    Given the problem instance $(N, M, V, G)$, with $|N| > \Delta(G)$, find an
    $\alpha$-approximate MMS allocation for sufficiently small $\alpha > 0$.

    \begin{enumerate}
        \item Scale valuations so that $v_i(M) = |N|$ for all $i \in N$.
        \item Repeatedly reduce the instance by allocating single items worth
            at least $\alpha$ to some remaining agent.
            Let $N'$ and $M'$ be the remaining agents and items at any stage,
            respectively. Rescale valuations to $v_i(M') = |N'|$ for all $i \in
            M'$ after each reduction.
        \item
            \begin{enumerate}
                \item If $|N'| = 1$, give the remaining agent $i$ an approximate
                    maximum weighted independent set on $G[M']$, weighted by
                    $v_i$, using the approximation algorithm of
                    \citet{Halldorsson:1997}.
                \item If $|N'| \ge 2$ and $G[M']$ is $\Delta(G)$-colorable,
                    perform \cref{alg:bag-filling} on any $\Delta(G)$-coloring
                    of $G[M']$.
                \item If $\Delta(G) \ge |N'| \ge 2$ and $G[M']$ is not
                    $\Delta(G)$-colorable, select an agent $i \in N'$. Create a
                    bundle $B$ of $i$'s least-valued item in each
                    non-$\Delta(G)$-colorable component of $G[M']$ and create a
                    $\Delta(G)$-coloring of $G[M'\setminus B]$. Scale the
                    valuations of $i$ so that $v_i(M' \setminus B) = |N'|$,
                    unless this results in $v_i(M') > |N'| + 1$, in which case
                    rescale to $v_i(M') = |N'| + 1$. Perform
                    \cref{alg:bag-filling} on $B$ and the $\Delta(G)$-coloring,
                    prioritizing other agents than $i$ in when breaking ties.
                \item If $|N'| > \Delta(G) \ge 1$ and $G[M']$ is not
                    $\Delta(G)$-colorable, perform \cref{alg:bag-filling} on a
                    ($\Delta(G) + 1$)-coloring of $G[M']$.
            \end{enumerate}
        \item Allocate remaining items without introducing conflicts in any
            bundle.
    \end{enumerate}
\end{algorithm}

\noindent
\Cref{alg:mms-polynomial} uses the same basic strategy as the approach in the
existence proof, except that it relies on colorings that may be found in
polynomial time. Note that instead of each special case in the bag filling step
of the algorithm, one could use a ($\Delta(G) + 1$)-coloring of $G[M']$. This
would, however, result in a reduction in the highest viable $\alpha$ by almost a
factor of $2$ in the worst case.

\begin{theorem}\label{thr:mms-approx-polytime}
    For a problem instance $(N, M, V, G)$, with $|N| > \Delta(G)$,
    \cref{alg:mms-polynomial} finds an $\alpha$-approximate MMS allocation in
    polynomial time, with $\alpha > 1/\Delta(G)$ when $\Delta(G) > 2$.
    Specifically,

    \[
        \alpha =
        \begin{cases}
            1/2 & \text{if } \Delta(G) = 1 \eqcomma \\
            3/7 & \text{if } \Delta(G) = 2 \eqcomma \\
            2/(\Delta(G) + 2) & \text{if } \Delta(G) > 2 \eqdot \\
        \end{cases}
    \]
\end{theorem}

\begin{proof}
    In order for the algorithm to run in polynomial time, it cannot perform the
    NP-hard calculation of $\mu_i$. Instead, normalization is used, providing an
    upper bound on the value. This is done by continuously scaling the values of
    each agent $i$, so that $v_i(M') = |N'|$ before, during and after step 2,
    where $M'$ and $N'$ are the remaining items and agents, respectively. By
    \cref{lem:normalization,lem:alpha-reduction}, this guarantees
    that $\mu_i \le 1$ in all parts of the algorithm, except for agent $i$ in
    step 3c. The equivalent validity of the rescaling of agent $i$'s valuations
    in step 3c will be covered later, together with the rest of the step.

    In step 3a, all agents, except one, have received one-item bundles. This
    means that for the remaining agent, the allocated items can at most have
    been taken from $|N| - 1$ of the $|N|$ bundles in any MMS partition of the
    original instance. Consequently, at least one of the bundles in each
    MMS partition exists in $G[M']$. Since a feasible bundle forms an
    independent set in the conflict graph, any maximum weighted independent set,
    $S$, of $G[M']$ weighted by this agent's valuations must fulfill $v_i(S)
    \ge \mu_i$. The approximation algorithm of \citet{Halldorsson:1997} has an
    approximation factor of $3/(\Delta(G) + 2)$, which for all $\Delta(G) \ge 1$
    is better than the limits of the theorem.

    In step 3b, \cref{lem:bag-filling} guarantees that $\alpha$ can be
    $|N'|/(2|N'| + \Delta(G) - 2)$, which when $\Delta(G) \le 2$ is at least
    $1/2$ and otherwise minimized when $|N'| = 2$, i.e., $\alpha = 2/(\Delta(G)
    + 2)$. Similarly, in step 3d, the function is minimized when $|N'| =
    \Delta(G) + 1$, allowing for $\alpha = (\Delta(G) + 1)/(3\Delta(G) + 1)$,
    which is $1/2$ for $\Delta(G) = 1$, $3/7$ for $\Delta(G) = 2$ and otherwise
    at least $2/(\Delta(G) + 2)$, satisfying the guarantees of the theorem.

    For step 3c we must show that
    \begin{inl}
        \item $G[M' \setminus B]$ is $\Delta(G)$-colorable,\label{item:3c-i}
        \item the rescaling of agent $i$'s valuation guarantees that $\mu_i \le
            1$, and\label{item:3c-ii}
        \item each remaining agent receives a bundle that satisfies the
            requirements set forth in the theorem.\label{item:3c-iii}
    \end{inl}

    For \ref{item:3c-i}, Brooks' theorem \cite{Brooks:1941} tells us that each
    non-$\Delta(G)$-colorable component is either the complete graph of
    $\Delta(G) + 1$ vertices, $K_{\Delta(G) + 1}$, or, if $\Delta(G) = 2$, a
    cycle of odd length, $C_{2k + 1}$. Removing a single item from the component
    will result in an induced subgraph of, respectively, $K_{\Delta(G)}$ and a
    path. Each of these can easily be verified to be $\Delta(G)$-colorable for
    $\Delta(G) \ge 2$. Subsequently, $G[M' \setminus B]$ is
    $\Delta(G)$-colorable.

    For the rescaling of agent $i$'s valuations, \ref{item:3c-ii}, let $A =
    \langle A_1, A_2, \dots, A_{|N|} \rangle$ be any of $i$'s MMS partitions of
    the original instance. Since $G[M']$ is not $\Delta(G)$-colorable, either a
    subset of the bundles in $A$ that contained items allocated in step 2 must
    contain an item from each non-$\Delta(G)$-colorable component or at least
    $\Delta(G) + 1 \ge |N'| + 1$ bundles remain untouched in step 2. In the
    first case, the remaining value is at least $|N'|\mu_i + v_i(B)$, i.e.,
    $v_i(M' \setminus B) \le |N'|$ guarantees that $\mu_i \le 1$. In the latter,
    the remaining value is at least $(\Delta(G) + 1)\mu_i$, and $v_i(M') \le
    |N'| + 1 \le \Delta(G) + 1$ guarantees that $\mu_i \le 1$. Scaling the
    valuations to the minimum of these, guarantees that either the value of the
    bundles in the $\Delta(G)$-coloring is $|N'|$ or the total value is $|N'| +
    1$. In the first case, \cref{lem:bag-filling} allows $\alpha = 2/(\Delta(G)
    + 2)$ for $i$, while in the latter it allows $\alpha = 3/(\Delta(G) + 2)$
    for $i$.  In either case, agent $i$'s bundle is compatible and $\alpha$ is
    at least as high as the one given in the theorem.

    For the remaining agents, \ref{item:3c-iii}, the bag filling algorithm is
    performed with $\Delta(G) + 1$ bundles and a total value of $|N'|$. This
    would by \cref{lem:bag-filling} result in $\alpha = 2/(\Delta(G) + 3)$;
    however, since $i$ is never selected when breaking ties, either all other
    agents are given a bundle before $i$ or the bundle given to agent $i$ is
    worth less than $\alpha$. Consequently, we need to select $\alpha$ so that
    the following holds:
    \begin{align*}
        (\Delta(G) + 1)\alpha &\le |N'| - 2(|N'| - 2)\alpha - \alpha \\
        \alpha&\le 2/(\Delta(G) + 2)
    \end{align*}
    \cref{lem:alpha-reduction} guarantees that step 4 completes the
    $\alpha$-approximate MMS allocation.

    It can easily be verified that steps 1, 2, 3a and 4 can all be completed in
    polynomial time in the number of agents and items. Checking that
    $G[M']$ is $\Delta(G)$-colorable is equivalent to checking if the induced
    subgraph contains either $K_{\Delta(G) + 1}$ or if $\Delta(G) = 2$ a
    component that forms an odd cycle. Each of these checks can be performed in
    polynomial time. A $(\Delta(G) + 1)$-coloring can greedily be constructed in
    polynomial time and a polynomial-time algorithm of \citet{Lovasz:1975} can
    be used to construct a $\Delta(G)$-coloring in a $\Delta(G)$-colorable
    graph. Thus, steps 3b, 3c and 3d can be completed in polynomial time and the
    algorithm will find an $\alpha$-approximate MMS allocation in polynomial
    time.
\qed
\end{proof}

\noindent
\cref{alg:mms-polynomial} relies on ($\Delta(G) + 1$)- and $\Delta(G)$-colorings
in order to achieve its guarantee, as at least one of them can be constructed in
polynomial time for any type of graph. Many graphs admit colorings using fewer
colors. Optimal colorings can generally not be found in polynomial time.
However, for some restricted graph classes we can find optimal colorings in
polynomial time and for other there exists methods for finding colorings within
a certain factor of optimum. For these classes, using the simpler colorings
(i.e., with fewer colors) will improve the guarantees of \cref{alg:bag-filling},
which would improve the possible approximation factor that can be used in a
polynomial-time algorithm. One such special case, is bipartite graphs, for which
a 2-coloring can be found in polynomial time. \Cref{cor:bipartite} shows the
possible approximation factor in this case, which is almost as good as our
earlier theoretical lower bound.

\begin{corollary}\label{cor:bipartite}
    For a problem instance $(N, M, V, G)$, where $G$ is a bipartite graph and
    $|N| > \Delta(G)$, a $1/2$-approximate MMS allocation can be found in
    polynomial time.
\end{corollary}

\begin{proof}
    When $G$ is bipartite, we do not need to rely on a $\Delta(G)$-coloring, as
    in \cref{alg:mms-polynomial}, in order to find an approximation in
    polynomial time. Instead, a $2$-coloring can greedily be found in polynomial
    time. If the $\Delta(G)$-coloring is replaced by a $2$-coloring, the proof
    of \cref{thr:mms-approx-polytime}, without consideration for
    $\Delta(G)$-colorablility and ($\Delta(G) + 1$)-colorings, guarantees an
    $\alpha$ of $1/2$.
\qed
\end{proof}

\noindent
While the use of better colorings can improve the value of $\alpha$ for specific
graph classes, there are still quite strong limitations on how well
\cref{alg:bag-filling} can perform on a coloring of the graph. We
will, for the straightforward approach used, be limited by an $\alpha$ of around
$1/\chi(G)$---and generally worse, because of the problems of finding a
$\chi(G)$-coloring. Using existing algorithms for more complex types of
valuation functions with no conflicting items, we can guarantee a constant
approximation factor in polynomial time for certain classes of graphs.%
\footnote{Post-publication note: The results of \citet{Li:2021} can be used to
find $(11/30 - \epsilon)$-approximate MMS allocations ($\epsilon > 0$) in
polynomial time when the maximum weighted independent set problem can be solved
in polynomial time on the graph G.}

\begin{proposition}\label{prop:xos-reduction}
    For a problem instance $(N, M, V, G)$ with $|N|>\Delta(G)$, a
    $1/8$-approximate MMS allocation may be found in polynomial time if the
    maximum weighted independent set problem can be solved in polynomial time on
    $G$ and all induced subgraphs of $G$.
\end{proposition}

\begin{proof}
    To find a polynomial-time algorithm, we will show a reduction to an
    unconstrained instance with fractionally subadditive (XOS) valuations.
    \citet{Ghodsi:2018} showed that for an instance of this type, a
    polynomial-time algorithm exists for $1/8$-approximate MMS allocations given
    a polynomial-time \textit{demand oracle} and \textit{XOS oracle}.

    Recall that a set function $f: 2^M \to \mathbb{R}_{\ge 0}$ is XOS if it can
    be represented as a finite set of additive functions $F = \{f_1, f_2, \dots,
    f_\ell\}$, such that $f(X)$ is the maximum of all the additive functions
    applied to $X$, i.e., $f(X) = \max_{i = 1}^{\ell} f_i(X)$. We will show that
    an oracle that finds the maximum weight $f(S)$ of an independent set in
    $S\subseteq M$ is XOS. For a graph $G$, let $\Ind(G)$ be the set of
    all the independent sets of $G$. For each $I \in \Ind(G)$, let $f_I : 2^M
    \to \mathbb{R}_{\ge 0}$ be an additive function, with values for $m \in M$
    given by
    \[
        f_I(\{m\}) =
        \begin{cases}
            w(m) & \text{if } m \in I\,;
            \\ 0 & \text{otherwise.}
        \end{cases}
    \]
    For a set of vertices $S \subseteq M$, $f_I$ finds the sum of the weight of
    the vertices in $S \cap I$. Let $F = \{f_I : I \in \Ind(G)\}$. Because
    $\Ind(G)$ contains all possible independent sets of $G$, the XOS function
    over $F$, $f'$, will for any set of vertices $S \subseteq M$ be maximized by
    $f_I$ for any maximum weighted independent set, $I$, on $G[S]$. Thus, $f =
    f'$ and the oracle is XOS. Consequently, XOS valuations may be calculated
    for each agent $i$ by using a maximum weighted independent set oracle on the
    graph using agent $i$'s valuations as vertex weights.

    We wish to use the $1/8$-approximation algorithm for unconstrained MMS for
    instances with XOS valuations to find $1/8$-approximate MMS allocations for
    instances with additive valuations and conflicting items. In order for this
    to be possible, we need to show that feasible allocations in either setting
    can be converted to feasible allocations in the other with at least the same
    value. If this is possible, then the MMS guarantee is the same in both
    settings and any $\alpha$-approximate MMS allocation in one setting can be
    converted to an $\alpha$-approximate MMS allocation in the other setting.

    Let $A = \langle A_1, A_2, \dots, A_{|N|} \rangle$ be an allocation in the
    XOS setting. For each agent $i$, a maximum weighted independent set, $S_i$,
    of $A_i$, with respect to agent $i$'s valuations, will have the same value
    as $A_i$. Consequently, $S = \langle S_1, S_2, \dots, S_{|N|} \rangle$ is an
    allocation, possibly partial, where each agent receives the same value as in
    $A$. Since each $S_i$ forms an independent set in $G$, $S$ is also a
    feasible partial allocation with respect to the conflict graph, where each
    agent receives the same value as in the XOS setting. Without losing value,
    the unallocated items may be allocated to arbitrary agents without
    conflicts, as $|N| > \Delta(G)$. As a result, all allocations in the XOS
    setting can be converted to ones in the conflicting items setting with at
    least the same value.

    Let $A = \langle A_1, A_2, \dots, A_{|N|} \rangle$ be a feasible allocation
    in the conflicting items setting. Then each bundle $A_i$ forms an
    independent set of $G$ and $A_i \in \mathcal{I}$. Consequently,
    $f_{A_i}(A_i) = v_i(A_i)$ and bundle $A_i$ has at least the same value in
    the XOS setting.

    In order for the $1/8$-approximation algorithm of \citeauthor{Ghodsi:2018}
    to be usable in polynomial time, we must show that the reduction and
    subsequent conversion of the $\alpha$-approximate MMS allocation can be
    performed in polynomial time. The latter must be true, as it consists of
    solving the maximum independent set problem $|N|$ times in addition to
    reallocating the set of unallocated items, which can both be performed in
    polynomial time given our assumptions.

    In order for the reduction to be performable in polynomial time, the
    \textit{demand oracle} and \textit{XOS oracle} used in the algorithm of
    \citeauthor{Ghodsi:2018} must both be creatable and queryable in polynomial
    time. The latter is an oracle that given a set of items, $S$, and the XOS
    function of an agent, $f_i$, provides a representation of the value each
    item in $S$ contributes in the $f \in F_i$ that is maximized for $S$. The
    construction of $f_i$ implies that $f = f_I$, where $I$ is a maximum
    weighted set of $G[S]$ with respect to agent $i$'s valuations. An oracle
    that finds $I$ and provides a representation of $f_I$ is therefore an XOS
    oracle. Both of these operations can be performed in polynomial time, given
    our assumptions.

    The demand oracle is an oracle that, given a list of prices $\langle p_1,
    p_2, \dots, p_m \rangle$ and the XOS function of an agent, $f_i$, finds a
    set of items, $S$, that maximizes $f_i(S) - \sum_{s \in S} p_s$. In the
    algorithm, the prices are non-negative, which implies that $v_{is} - p_s \le
    v_{is}$. Let $S'$ be the possibly empty set of items, where $v_{is} < p_s$.
    Then any maximum weighted independent set, $S$, in $G[M \setminus S']$ with
    weights $v_{is} - p_s$ maximizes $f_i(S) - \sum_{s \in S} p_s$. To see that
    this is true, let $S$ be any set that maximizes the function. Any item in
    $S'$ would provide a negative contribution and cannot be in $S$, i.e., $S
    \subseteq M \setminus S'$. If $S$ is an independent set, then $f_i(S) -
    \sum_{s \in S} p_s = \sum_{s \in S} (v_{is} - p_s)$, and consequently if $S$
    is an independent set, it is a  maximum weighted independent set in $G[M
    \setminus S']$ with weights $v_{is} - p_s$. Now, we only need to show that
    for any maximum that is not an independent set, there exists an independent
    set with at least the same value. Assume that $S$ is not an independent set,
    and let $S^*$ be a maximum weighted independent set of $S$ with respect to
    the valuations of $i$. Then $f_i(S) = v_i(S^*)$ and $f_i(S) - \sum_{s \in S}
    p_s \le f_i(S^*) - \sum_{s \in S^*} p_s$.  Consequently, finding a maximum
    weighted independent set maximizes the function. This can be done in
    polynomial time and thus, the demand oracle can be queried in polynomial
    time.
\qed
\end{proof}

\noindent
The maximum weighted independent set problem is NP-hard, which limits the
usefulness of \cref{prop:xos-reduction}. However, there are several graph
classes where the problem is solvable in polynomial time; two noteworthy
examples are bipartite graphs and claw-free graphs~\citep{Minty:1980}. While we
for bipartite graphs already know how to find $1/2$-approximate MMS allocations
in polynomial time, there is no restriction on the possible values for $\chi(G)$
in claw-free graphs.\footnote{For example, a complete graph is claw-free and has
$\chi(G) = \Delta(G) + 1$.} As a result, the reduction to XOS valuations could
provide major improvements for dense, claw-free conflict graphs.

While we for some graphs cannot use the results of \cref{prop:xos-reduction} to
find $1/8$-approximate MMS allocations in polynomial time, the reduction to XOS
valuations holds nonetheless. This provides an alternative existence proof for
MMS approximation, as \citet{Ghodsi:2018} showed the existence of
$1/5$-approximate MMS allocations for XOS valuations. Note that
\citeauthor{Ghodsi:2018} also showed similar, but better, results for the more
restrictive setting of submodular valuation functions. It might tempting to
think that the maximum weighted independent set oracle is submodular; however,
it is not, as is evident from \cref{exp:submodular}.

\begin{example}\label{exp:submodular}
    Consider a problem instance consisting of 4 items, where the valuation
    function of agent $i$ is given by $v_{i1} = 5, v_{i2} = 3, v_{i3} = 5,
    v_{i4} = 3$. Let the conflict graph be as follows:

    
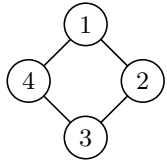
\begin{center}
        \begin{tikzpicture}[graph]
            \draw
                (0.75, 1.5) node (1) {1}
                (1.5, 0.75) node (2) {2}
                (0.75, 0) node (3) {3}
                (0, 0.75) node (4) {4}
                ;

            \draw
                (1) edge (2)
                (2) edge (3)
                (3) edge (4)
                (4) edge (1)
                ;
        \end{tikzpicture}
    \end{center}
    Let $f$ be the function that given a set of items finds the weight of the
    maximum weighted independent set using the valuations of the agent as vertex
    weights. If we let $X = \{1, 2, 4\}$ and $Y = \{2, 3, 4\}$, then $f(X) +
    f(Y) = 12 < 16 = f(X \cup Y) + f(X \cap Y)$. In other words, $f$ is not
    submodular.
\end{example}

\section{Conflicting Items in Practice}
\label{sec:practice}

We are interested in determining how fairness is affected, and the extent to
which existing tools and formalisms still apply, with item conflicts.
Specifically, when imposing item conflicts:

\begin{rqs}[itemsep=0.75ex]
    \item To what extent do fair allocations (EF1, MMS) exist?
        \label{rq:exist}
    \item How is the fairness (MMS, PROP) of random allocation affected?
        \label{rq:rand}
    \item To what extent does MNW imply fairness (EF1, MMS)?
        \label{rq:mnw}
\end{rqs}
First of all, we wish to know to what extent fair allocation is actually
possible in this new setting (\ref{rq:exist}). We know EF1 is not guaranteed in
this setting, and that MMS may not be achievable even in the unconstrained
case~\citep{Kurokawa:2016}, but will the fairness guarantees go down to the
point where these properties become the exception rather than the rule? And if
MMS were to be unattainable for some instance, was it attainable \emph{without}
the item conflicts, or could the opposite be true, because of declining
individual maximin shares?

We also wish to look at the \emph{prevalence} of fair allocations
(\ref{rq:rand}). With item conflicts, many allocations are no longer feasible,
so utility will tend to decrease. However, whether \emph{fairness} decreases is
not a given. For example, although individual maximin shares will generally be
lower in any setting with additional constraints, the degree of attainable MMS
approximation may very well go \emph{up}. And although one may lose the fairest
allocations, it is quite possible that one loses many more \emph{unfair}
allocations; each item conflict, for example, prevents some degree of hoarding,
as no agent can hold both items. To examine this effect empirically, we apply a
random allocation procedure, as described in \cref{prop:random}, studying the
effect on the approximation of both MMS and proportionality, the latter to
separate the effect of forced distribution of items from the lowering of
individual maximin shares. This question is of interest in its own right,
describing inherent properties of the problem. An answer might, however, also
shed some light on the relative hardness of finding fair allocations, e.g.,
through randomized or heuristic procedures, with and without item conflicts.

Finally, we look at the impact on the maximum Nash welfare (MNW), inspired by
the work of \citeauthor{Caragiannis:2019} (\ref{rq:mnw}). They show that MNW is
a useful tradeoff between efficiency and fairness, in the unconstrained
allocation setting~\citep{Caragiannis:2019}. Is this still the case with item
conflicts? For one thing, we look for a decrease in MNW, which might indicate
either lower efficiency or fairness, or both. And while MNW still implies Pareto
optimality when item conflicts are introduced, it no longer implies EF1 in
general, and it is uncertain to which extent its approximation of MMS is
preserved, so we explore the relationship between MNW and both of these
properties.

\subsection{Experimental Setup}

To address our empirical questions, we generated a collection of random
instances, and found randomized allocations, MMS allocations, MNW allocations
with and without EF1, all with and without item conflicts.

\paragraph{Problem instances.} A central issue in generating instances is the
choice of graph models. We selected three of the most popular and well-studied
models of real-world graphs as our constraints:

\begin{stmts}[widest=iii, itemsep=0.75ex]
    \item The Erdős--Rényi
        model, where each edge is present with
        probability $p$~\citep{Bollobas:2001};
    \item The Barabási--Albert model, where a small graph is extended by
        preferential attachment~\citep{Albert:2002}; and
    \item The Watts--Strogatz model, where the edges of a regular ring lattice
        are randomly rewired with a certain probability~\citep{Watts:1998}.
\end{stmts}

\noindent
Before producing an instance for a given graph type, all parameters were
selected uniformly at random, in the ranges shown in \cref{tab:params}. The
limits for $n$ and $m$ are based on real-world data from
\citet{Caragiannis:2019}, where the largest $n$ observed was~\num{10}, and the
average ratio $m/n$ was approximately~\num{3}. The upper limit for $d$ is based
on the assumption of \citeauthor{Watts:1998} that $d\ll m$~\citep{Watts:1998}.

\begin{table}
\caption{The ranges of the randomly selected graph parameters}
\label{tab:params}
\begin{tabular}{@{}llll@{}}
\toprule
\textbf{Par.} & \textbf{Range} & \textbf{Description} & \textbf{Model} \\
\midrule
$n$ & $2,3, \dots, 10$ & The number of agents & All \\
$m$ & $2n,2n\!+\!1, \dots, 4n$ & The number of items (vertices) & All \\
$p$ & $(0, 1)$ & Edge probability & Erdős--Rényi \\
$k$ & $1, 2, \dots, m$ & Initial size and connection degree & Barabási--Albert \\
$d$ & $2, 4, \dots, m/2$ & Average degree & Watts--Strogatz \\
$\beta$ & $[\mkern1mu0, 1)$ & Rewiring probability & Watts--Strogatz \\
\bottomrule
\end{tabular}
\end{table}

Once the parameters were set, a random graph $G$ was generated. If the graph had
no edges, it was discarded, as this would merely be an instance of the ordinary
allocation problem.\footnote{This only applies to the Erdős--Rényi case.} If
$\Delta(G) \geq n$, the graph was also discarded; for such instances, an
allocation may not be feasible, and many results and methods, including the
randomized algorithm of \cref{prop:random}, do not apply. The process was
repeated until we had \num{5000} graphs of each kind for which $n\leq\CC(G)$.
This was done to have enough data to study the prevalence of EF1 in cases where
it is not immediately implied by \cref{prop:ef1}. Finally, for each instance, a
valuation was created by randomly dividing approximately~\num{1000} points among
the items, for each agent, in line with the value specification mechanism of
\citeauthor{Caragiannis:2019}
More specifically, each item was given a random real value, and the sum for each
agent was scaled to~\num{1000}. Finally, individual values were rounded, as the
mixed-integer program used to find MNW requires integral
valuations~\citep{Caragiannis:2019}.

\paragraph{Implementation.}

The experiments were implemented in the Julia programming language, version
1.5.3~\citep{Bezanson:2017}, using the LightGraphs
package~\citep{Bromberger:2017} for handling graphs, and the JuMP
package~\citep{Dunning:2017} with Gurobi 9.1.1~\citep{Gurobi:2021} as the
backend for solving mixed-integer linear programs (MIPs).\footnote{The source
code and raw experimental results are available as ancillary files for the
preprint of this paper at \url{https://arxiv.org/abs/2104.06280}.} For the randomized
allocation, \num{1000} trials were performed for each instance, and the average
recorded. Gurobi was run with a timeout of \SI{5}{\min} (on an 8-core Intel
i9-9900K at \SI{3.60}{\GHz}).

For finding individual maximin shares, a straightforward maximin MIP was used,
with $n$ clones of the given agent. The MMS allocations were then found by
another maximin MIP, where agents' values were divided by their maximin shares.
These maximin shares were also used to find the proportion of MMS for randomized
allocation and MNW. The MNW allocations were found by a MIP based on the one
described by \citet{Caragiannis:2019} (adapted to permit varying maximum bundle
values), and MNW with EF1 was found by the same program, with added constraints
requiring EF1. In all cases, item conflicts were handled by adding the necessary
constraints to the relevant MIPs.

\subsection{Experimental Results}

\begin{table}
\caption{Summary of the generated instances}%
\label{tab:instances}%
\input{insttable.tex}

\end{table}

A summary of the instances is given in \cref{tab:instances}. While running the
experiments, \ndropped{} instances (\SI{\pctdropped}{\percent}) timed out when
solving the mixed-integer programs for MMS, and were dropped from any MMS-based
calculations. The number of remaining instances are listed in the MMS column.
Following this are averages for the number of agents ($n$), items ($m$) and
edges ($|E|$), as well as for the maximum degree ($\Delta$) and largest
component ($\CC$).
Beyond timeouts, another issue is the required precision to compute MNW, which
may be quite high, even for a modest number of agents. When the precision is not
available, the result may not be exactly MNW, but will still---with much laxer
precision requirements---be Pareto-optimal and, for the unconstrained case,
EF1~\citep[cf.][]{Caragiannis:2019}. For our instances, only
\SI{\pctmnwprec}{\percent} had the sufficient precision, and were thus
guaranteed to find an MNW allocation, and not some close approximation. This
means that our results may in some sense be seen as an empirical evaluation of
the mixed-integer program of \citet{Caragiannis:2019} for finding MNW, rather
than of the MNW itself. As this is currently the only feasible approach for
finding MNW solutions, this still tells us something about the usefulness of MNW
in practice. And although the solutions found might, in principle, lose some of
the power of a guaranteed MNW solution, it seems unlikely that they would do
\emph{better} on any fairness criteria, at least in any systematic way; thus
positive results could still be taken as support for the use of MNW.

\paragraph{The existence of fairness.}
Without constraints, EF1 allocations always exist, and all but the rarest
instances have MMS allocations; we wish to know to what extent this holds true
also under item conflicts (\ref{rq:exist}). Under item conflicts, we have
established that EF1 is not guaranteed for any graph when $n\leq\Delta(G)$
(\cref{prop:noef1}), and is otherwise guaranteed in some cases
(\cref{prop:ef1,prop:pathef1,prop:equit}) but not in others
(\cref{prop:non-existence-ef1}). The question is how common it is in practice.
Similarly, we have shown that many approximations to MMS are guaranteed to
exist, and although MMS is not guaranteed in
general, might their existence still be the rule, despite having conflicting
items? Our results are easily summed up in the affirmative: Every instance had
both an EF1 allocation and an MMS allocation (also in the unconstrained case).

\paragraph{Expected fairness.}

When allocating items randomly to agents, one would expect a fairly even
distribution, and introducing item conflicts does not affect this expectation
too much on an individual level (cf.\@ \cref{prop:random}). We wish to examine
the effects on the minimum as well, i.e., how closely we approximate MMS and
proportionality (\ref{rq:rand}). For our random instances, random allocation
disregarding the conflicts (i.e., a lottery) achieved, on average, an
\randmmsalphawocon-approximate MMS allocation. With conflicts, however, we got
an average approximation ratio of \randmmsalpha. A possible cause for the
increase is that in the maximin problem being solved, values are normalized by
the individual maximin shares of the agents, and these will generally go down
with additional constraints. However, a similar increase is seen in the
approximation ratio for proportionality (from \randpropalphawocon\@ to
\randpropalpha), which would seem to indicate that the conflicts enforce a
certain level of distribution of the items, with a forced increase in fairness.
The distribution of approximation ratios for MMS is shown in
\cref{fig:randalpha}.

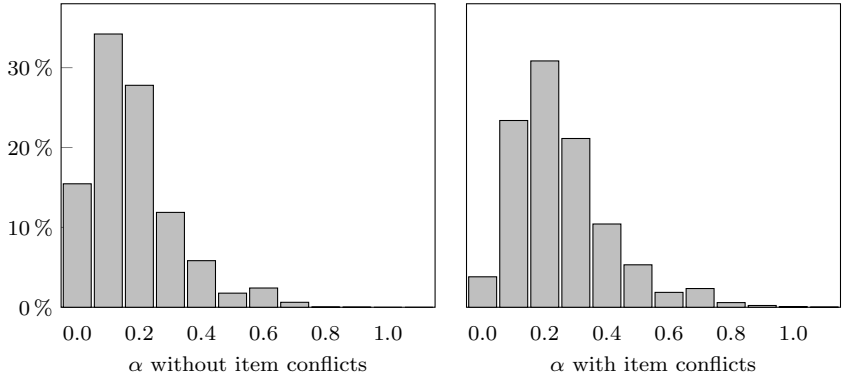
\begin{figure}
\centering
\begin{tikzpicture}
    \begin{axis}[
        small,
        ybar,
        xtick=data,
        bar width=.9,
        enlarge x limits={abs=.52},
        xticklabels from table={randalphabars.tsv}{e},
        xtick style={draw=none},
        xtick pos=bottom,
        ytick pos=left,
        ymin=0,
        ymax = .38,
        yticklabel={\phantom{1}\pgfmathparse{\tick*100}\pgfmathprintnumber{\pgfmathresult}\,\%},
        xlabel={$\alpha$ without item conflicts},
        x tick label style={inner xsep=0pt},
    ]
    \addplot[fill=lightgray] table[x=x, y=ywo] {randalphabars.tsv};
    \end{axis}
\end{tikzpicture}
\quad
\begin{tikzpicture}
    \begin{axis}[
        small,
        ybar,
        xtick=data,
        bar width=.9,
        enlarge x limits={abs=.52},
        xticklabels from table={randalphabars.tsv}{e},
        xtick style={draw=none},
        xtick pos=bottom,
        ytick align=outside,
        ytick pos=left,
        ymin=0,
        ymax=.38,
        ymajorticks=false,
        xlabel={$\alpha$ with item conflicts},
        x tick label style={inner xsep=0pt},
    ]
    \addplot[fill=lightgray] table[x=x, y=y] {randalphabars.tsv};
    \end{axis}
\end{tikzpicture}
\caption{The distribution of the approximation ratio $\alpha$ for MMS of a
random distribution, with randomized tie breaking in the case of item
conflicts (\cref{alg:random}). The bin labels indicate the lower inclusive
limits. The item conflicts are seen to force a shift toward MMS}
\label{fig:randalpha}
\end{figure}

\paragraph{The fairness of MNW.}

As opposed to the proportion of MMS for randomized allocation, the maximum Nash
welfare will never increase when adding item conflicts, simply because the
original is then a relaxation, with a higher possible optimum. The question is
how much this impacts its usefulness as a tool for fair allocation
(\ref{rq:mnw}). For one thing, we know it can no longer guarantee EF1, even with
$n>\Delta(G)$~(\cref{prop:mnwnotef1}). It turns out, however, that it is still
very close, with \SI{\efonefrommnw}{\percent} of the MNW solutions for
conflicting items being EF1. Conversely, adding a requirement of EF1 reduces the
average MNW by \SI{\mnwchange}{\percent}.

The reduction in proportion of MMS was not too substantial either. On average,
the MNW solutions for our instances were \mnwmmsalphawocon-MMS, and adding item
conflicts reduced this to \mnwmmsalpha-MMS (i.e., on average well above full MMS
in both cases). The proportion of cases where MMS was achieved fell from
\SI{\mnwmmswocon}{\percent} to \SI{\mnwmms}{\percent}, as illustrated in
\cref{fig:mnwalpha}.

\begin{figure}
\centering
\begin{tikzpicture}
    \begin{axis}[
        small,
        ybar,
        xtick=data,
        bar width=.9,
        enlarge x limits={abs=.52},
        xticklabels from table={mnwalphabars.tsv}{e},
        xtick style={draw=none},
        xtick pos=bottom,
        ymin=0,
        ymax=1.11,
        ytick pos=left,
        point meta={y*100},
        nodes near coords={\scriptsize\pgfmathprintnumber\pgfplotspointmeta\%},
        nodes near coords style={/pgf/number format/precision=3,
             /pgf/number format/fixed},
        yticklabel={\pgfmathparse{\tick*100}\pgfmathprintnumber{\pgfmathresult}\,\%},
        xlabel={$\alpha$ without item conflicts},
    ]
    \addplot[fill=lightgray] table[x=x, y=ywo] {mnwalphabars.tsv};
    \end{axis}
\end{tikzpicture}
\quad
\begin{tikzpicture}
    \begin{axis}[
        small,
        ybar,
        xtick=data,
        bar width=.9,
        enlarge x limits={abs=.52},
        xticklabels from table={mnwalphabars.tsv}{e},
        xtick style={draw=none},
        xtick pos=bottom,
        ytick align=outside,
        ymin=0,
        ymax=1.11,
        ymajorticks=false,
        point meta={y*100},
        nodes near coords={\scriptsize\pgfmathprintnumber\pgfplotspointmeta\%},
        nodes near coords style={/pgf/number format/precision=3,
             /pgf/number format/fixed},
        xlabel={$\alpha$ with item conflicts},
    ]
    \addplot[fill=lightgray] table[x=x, y=y] {mnwalphabars.tsv};
    \end{axis}
\end{tikzpicture}
\caption{The proportion of MMS attained by MNW before and after adding item
conflicts}
\label{fig:mnwalpha}
\end{figure}
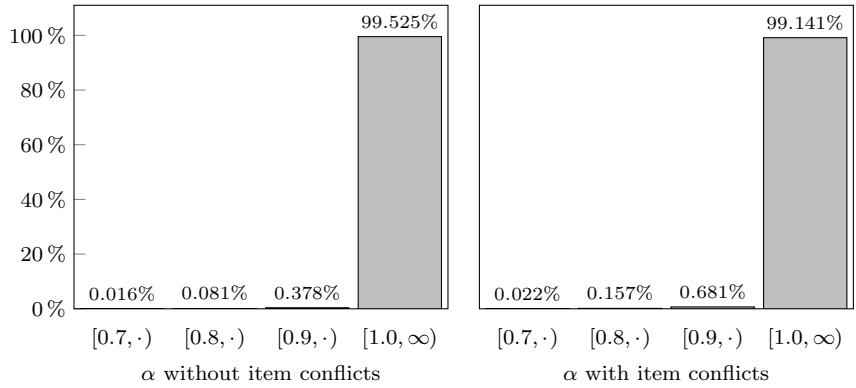

\section{Discussion}

The main purpose of this paper was to explore how introducing item conflicts
affects the existence, prevalence and implications of fairly common fairness
criteria. From a general theoretical standpoint, item conflicts seem to result
in a certain reduction in fairness guarantees. At least, both the existence of
EF1 and the guarantee that MNW leads to EF1 no longer hold. However, from a
practical point of view, the reductions in guarantees seem to be minor and for
some cases---non-existence of EF1---do not occur in our randomly generated
instances.

It is interesting to see that while our theoretical results show that MNW no
longer leads to EF1, in practice this only occurs in a handful of instances.
Additionally, MNW seems to produce only slightly worse MMS approximations than
in the unconstrained setting, still mostly resulting in MMS allocations. While
no longer theoretically backed to the same extent as in the unconstrained setting,
in practice MNW still seems to mostly provide similar benefits as those
advocated by \citet{Caragiannis:2019}, providing a practical tradeoff between
efficiency and fairness.

For the non-existence of EF1, the experiments indicate that this most likely
only very rarely occurs. This might to a certain degree be indicative of the
extent to which there exist combinations of conflict graphs and number of
agents such that we can construct instances without EF1 allocations. An insight
that might help for the earlier mentioned open problem about which
combinations of conflict graphs and the number of agents EF1 always exist for.

Both the non-existence proofs and the relative dearth of counter-examples in
practice for both EF1 and the implication from MNW to EF1 tell us something
about the prevalence of these properties. They do not, however, tell the full
story about the probability of occurrence of non-existence. It might be that
instances generated in our experiments provide types of instances where the
non-existence is either likely or unlikely to occur in relation to the overall
likelihood for all instances. An exploration of the probability with which each
property is expected to occur would provide a better insight into the usefulness
of the properties in real-world settings. Note that similar studies have been
performed for envy-freeness in the unconstrained setting \cite[see,
e.g.,][]{Manurangsi:2018}.

For all but some special cases, there remains a large gap between our lower
bound for the theoretical approximation guarantees of MMS and what our
poly\-nomial-time algorithms can guarantee. An interesting continuation of the
research would be looking into improvements of the guarantees, both in general
and for polynomial-time algorithms, especially considering that all instances in
the experiments admit an MMS allocation. Conversely, it would not be unlikely
that some, not too great, upper bound on polynomial approximation guarantees
exist (unless P$\null=\null$NP), because of the close relation to hard graph
problems and the upper bound for the closely related problem variant of
\citet{Chiarelli:2020}.

A natural extension to conflicting items that may be interesting to explore is
introducing additional conflicts between agents and items, as in the standard
weighted bipartite matching problem, where some agents simply cannot receive
certain items. For example, this would be useful in the real-world example where
items represent positions or roles in an organization. With agent--item
conflicts, it would be possible to limit certain agents from taking on specific
roles, either due to outside conflicts of interest or a lack of the require
skillset needed for the role. In the same way that allocations for conflicting
items form $n$-colorings of the conflict graph, the extension to item--agent
conflicts would have allocations that form list colorings---a fairly well-known
generalization of graph coloring, which has also been studied in the context of
equitable coloring~\citep{Kostochka:2003,Lih:2013}.

\paragraph{Conflicts of interest.}

The authors have no relevant financial or non-financial interests to disclose.

\paragraph{Publication.} Except minor adjustments by the authors, this preprint
has not undergone any post-submission improvements or corrections. The version
of record of this article is published in \emph{Autonomous Agents and
Multi-Agent Systems} and is available online at
\url{https://dx.doi.org/10.1007/s10458-021-09537-3}.

\bibliography{paper}

\end{document}

%% file: saved.tex
\tl_set:cn{ntotal}{18629}
\tl_set:cn{ndropped}{115}
\tl_set:cn{pctdropped}{0.006}
\tl_set:cn{pctmnwprec}{24.15}
\tl_set:cn{randmmsalpha}{0.30}
\tl_set:cn{randmmsalphawocon}{0.23}
\tl_set:cn{randpropalpha}{0.29}
\tl_set:cn{randpropalphawocon}{0.22}
\tl_set:cn{efonefrommnw}{99.88}
\tl_set:cn{mnwmmsalpha}{1.41}
\tl_set:cn{mnwmmsalphawocon}{1.43}
\tl_set:cn{mnwmms}{99.14}
\tl_set:cn{mnwmmswocon}{99.52}
\tl_set:cn{mnwchange}{1.52}

%% file: insttable.tex
\begin{tabular}{@{}lccccccc@{}}
\toprule
\textbf{Type} &
\textbf{Num.} &
\textbf{MMS} &
$\bm{n}$ &
$\bm{m}$ &
$\bm{|E|}$ &
$\bm{\Delta}$ &
$\bm{\CC}$ \\
\midrule

Erdős--Rényi & 8620 & 8583 & 6.4 & 18.1 & 17.4 & 3.6 & 11.4\\
Barabási--Albert & 5000 & 4956 & 7.7 & 19.8 & 20.5 & 6.0 & 19.8\\
Watts--Strogatz & 5009 & 4975 & 7.1 & 20.1 & 30.7 & 4.4 & 19.4\\
\bottomrule
\end{tabular}